\documentclass[journal]{IEEEtran}

\usepackage[T1]{fontenc}
\usepackage{amsmath,amssymb,amsfonts}
\usepackage{amsthm}
\usepackage{mathtools}
\usepackage{bm}
\usepackage{graphicx}
\usepackage[caption=false,font=footnotesize]{subfig}
\usepackage{cite}
\usepackage{url}
\usepackage{booktabs}
\usepackage{array}
\usepackage{textcomp}
\usepackage{xcolor}
\usepackage{balance}
\usepackage[ruled,vlined,linesnumbered]{algorithm2e}

\newtheorem{theorem}{Theorem}
\newtheorem{lemma}{Lemma}
\newtheorem{proposition}{Proposition}
\newtheorem{corollary}{Corollary}

\theoremstyle{definition}
\newtheorem{remark}{Remark}
\newtheorem{example}{Example}

\begin{document}

\title{Rate-Distortion Theory of Gaussian Mixtures via Posterior Label Sampling}

\author{Weichen Wang, \textit{Student Member, IEEE}, Deniz Gündüz, \textit{Fellow, IEEE}%
\thanks{This work received funding from the UKRI for the projects INFORMED-AI (EP/Y028732/1),  AI-R (ERC Consolidator Grant, EP/X030806/1) and NSF-EPSRC Neural Joint Source-Channel Coding (UKRI1463).}%
}

\maketitle

\begin{abstract}
We study the mean-squared error (MSE) rate-distortion function (RDF) of finite Gaussian mixture models (GMMs) in \(\mathbb{R}^d\), for which neither the RDF nor an optimal test channel is known in closed form. It can be sandwiched between the genie-aided conditional RDF $R_{\mathbf{X} | T}(D)$, which can be characterized in closed-form for a GMM,  and $R_{\mathbf{X} | T}(D)+I(\mathbf{X};T)$, where $T$ denotes the latent component label. We show that the upper bound can be achieved operationally by the proposed Gibbs-relabelling method (GRM), which conveys a label from the posterior $P_{T|\mathbf{X}}$ through \textit{channel simulation} at rate $I(\mathbf{X};T)$, followed by vector quantization of the resulting exactly Gaussian conditional sources by reverse water-filling. Using the Shannon lower bound (SLB), we prove that GRM attains the SLB, and hence the RDF, of a general finite mixture if and only if every component attains its own SLB. For a GMM this holds throughout the fully active interval \(D/d<\min_{k,m}\lambda_{k,m}\), where
\(\lambda_{k,m}\) are the eigenvalues of the component covariances, and for scalar GMMs the condition is also necessary. 

Since GRM cannot operate below the label cost $I(\mathbf{X};T)$, we propose a complementary stochastic scalar quantization (SSQ) scheme for low rates, which simulates an $M$-ary reconstruction index. We show that at low rates SSQ with $M\geq 2$ matches the Shannon RDF to first order:
\(R^{ }(D_{\max}-\delta)=\delta/(2\tau^2)+o(\delta)\), where $\tau^2$ is a  directional sub-Gaussian proxy of the mixture, and quantify the price of a finite codebook through a strict rate gap at high rates and a distortion floor whose joint scaling in $M$ and number of mixture components is characterized. 
Numerical results illustrate both regimes and the finite-blocklength cost of channel simulation.
\end{abstract}


\section{Introduction}
\label{sec:intro}

Finite Gaussian mixture models (GMMs) are among the most widely used models for multimodal continuous-amplitude signals.
They are used in speech, audio, and image compression
\cite{HedelinSkoglund2000,SubramaniamRao2003,Cheng2020GMM}, in
channel-state-information acquisition and feedback for
multiple-input multiple-output systems
\cite{Wen2015GMM,Turan2023GMMFeedback}, and in the modeling of multimodal
learned latent representations
\cite{Dilokthanakul2016GMVAE,Jiang2017VaDE,Zong2018DAGMM}.
In each of these applications, the mixture signal is ultimately quantized for storage or digital transmission, which raises the basic rate-distortion questions for GMMs: Can we characterize the mean-squared error (MSE) rate-distortion function (RDF) of a finite GMM, and which codes attain it? This paper gives explicit, if partial, answers to both questions for GMMs in $\mathbb{R}^d$, and when the arguments permit, for general finite mixtures.

Answers to both questions are well-known for standard Gaussian sources.
Shannon's source coding theorem characterizes the RDF as an infimum of mutual information \cite{Shannon1959}. Under MSE the infimum is given by reverse
water-filling over the eigenvalues of the covariance matrix
\cite{Kolmogorov1956,Berger1971,CoverThomas}, and it is achieved operationally by
a Karhunen-Lo\`eve transform followed by independent coding of the transform coefficients \cite{HuangSchultheiss1963}.
Although a GMM is Gaussian conditionally on its component label, applying these constructions component-wise is non-trivial, because the component label $T$ is a latent variable of the source $\mathbf{X}\in \mathbb{R}^d$, rather than an observation available to the encoder. Neither the differential entropy nor the RDF, $R(D)$, is known in closed-form for a GMM source. Both can be approximated numerically by the Blahut-Arimoto (BA) algorithm \cite{Blahut1972,Arimoto1972} by discretizing the source over an ever finer grid, or by the annealing approach of Rose \cite{Rose1994}, but a numerical value does not say anything about the code that attains it.

The role of the component label can be made more precise through the \textit{conditional RDF}, $R_{\mathbf{X} | T}(D)$, of source $\mathbf{X}$ when the label $T$ is available to both encoder and decoder. It can be shown that \cite{Gray1972CRD,Gray1973CRD} 
\begin{equation}\label{eq:intro}
    R_{\mathbf{X} | T}(D)\leq R(D) \leq R_{\mathbf{X} | T}(D)+I(\mathbf{X};T).
\end{equation}
For a GMM, the conditional RDF is explicit: given the label, each component is coded by reverse water-filling under a common water level. The gap between the two sides of \eqref{eq:intro}, $I(\mathbf{X};T)\leq H(T)$, is the cost of not observing T.

Classified and mixture-based vector quantizers in the literature all essentially implement this upper bound, but with the label first estimated and then conveyed losslessly, at a cost of $H(T)$, rather than $I(\mathbf{X};T)\leq H(T)$.
Classified vector quantization (VQ) \cite{RamamurthiGersho1986}, built on the composite-source
model and on the Linde-Buzo-Gray algorithm \cite{LindeBuzoGray1980},
assigns each block to a class by a deterministic classifier and applies a
class-specific codebook.
GMM-based VQ \cite{HedelinSkoglund2000,SubramaniamRao2003} replaces the
classifier by a fitted mixture and designs parametric codebooks from the
component densities.
Weighted universal transform coding
\cite{EffrosChouGray1994,Effros1999WUTC} transmits a class index obtained
by an encoder-side rate-distortion search and then applies a
class-matched Karhunen-Lo\`eve transform with a per-class water level.
PrismQuant \cite{PrismQuant2026} uses a hard maximum-a-posteriori (MAP) label with a single global water level,
and shows that its excess rate over the genie-aided conditional
RDF is of order \(H(T)/d\) in the vector dimension $d$. In each of these schemes, the label is a deterministic function of $\mathbf{X}$, and is conveyed at rate $H(T)$ rather than $I(\mathbf{X};T)$, and none of them is known to achieve $R(D)$ at any finite dimension.

Our main contribution is the observation that the upper bound in \eqref{eq:intro} can be achieved exactly, with the label communicated at rate $I(\mathbf{X};T)$, provided it is not estimated at the encoder from $\mathbf{X}$, but rather \emph{sampled} given $\mathbf{X}$. Let us explain what we mean by this: Suppose the encoder draws a sample $\widetilde{T} \sim P_{T|\mathbf{X}}$ and the decoder has the same sample. By Bayes' rule, the pair $(\widetilde{T},\mathbf{X})$ has the same joint law as the pair $(T,\mathbf{X})$; hence, conditioned on $\widetilde{T}=k$, the source is exactly $\mathcal{N}(\boldsymbol{\mu}_k,\boldsymbol{\Sigma}_k)$, and the component water-filling code applies as is. Now, communicating a sample from a prescribed distribution is known as the problem of channel simulation, or reverse channel coding. By the \textit{reverse Shannon theorem} \cite{1035117}, with unlimited common randomness independent of $\mathbf{X}$, this costs $I(\mathbf{X};T)$ bits per letter asymptotically. The Poisson functional representation \cite{LiElGamal2018} is an explicit construction at $I(\mathbf{X};T)+O(\log(I(\mathbf{X};T)+1))$ bits for a single letter; hence at $I(\mathbf{X};T)+O(\log n/n)$ bits per letter on a block of $n$ letters \cite{FlamichGunduzREC}. Cuff \cite{Cuff2013} characterized the asymptotic trade-off when common randomness is rate-limited. For a finite mixture, we only need the exact simulation of a finite kernel. We call the resulting scheme the Gibbs-relabeling method (GRM), since the label is resampled from its posterior, as in a Gibbs sampler, instead of being assigned by a classifier as in prior methods. 

Even though the optimal RDF is achievable by fully deterministic codes, one can also simulate the optimal test channel directly. Under output-distribution or realism constraints this even becomes necessary \cite{saldi2017asymptotic, 9998081, 11558509}. The novelty of GRM is what we simulate: instead of simulating the test channel, we simulate the posterior label distribution, while the continuous part of the test channel is realized by a Gaussian code. This is important from a practical implementation point of view as the complexity of channel simulation grows exponentially in the Rényi $\infty$-divergence between the conditional and the marginal output laws for rejection-type samplers \cite{10619339}, and there is no polynomial algorithm in the divergence unless $RP=NP$ \cite{agustsson2020universally}; therefore, direct simulation of the test channel becomes intractable as $D \rightarrow 0$. The label, on the other hand, can be simulated with limited complexity, requires common randomness of rate at most $H(T|\mathbf{X})\leq \log K$, and does not depend on $D$, so a single label-simulation code can serve every distortion level and rate control can be applied solely through the water level of the Gaussian stage. Since the conditional source is Gaussian, standard tools for Gaussian sources (transform coding, reverse water-filling and high-resolution quantization theory) apply directly.

We employ the Shannon lower bound (SLB) \cite{Shannon1959,Berger1971} to determine when the upper bound in \eqref{eq:intro} is tight, and hence, when GRM achieves the RDF.
Under MSE the SLB asserts
\(R^{ }(D)\ge h(\mathbf{X})-\tfrac{d}{2}\log(2\pi e\,D/d)\), and holds with equality
if and only if $\mathbf{X}$ is the sum of an $\mathcal{N}(0,(D/d)I)$ vector and an independent random vector \cite{Berger1971}.
A Gaussian vector therefore meets the SLB exactly when its reverse water-filling is fully active; that is, when $D/d$ is below the smallest eigenvalue of its covariance. A scalar finite GMM meets it when
\(D<\min_k\sigma_k^2\), since each component, hence the mixture, is an \(\mathcal{N}(0,D)\)-convolution. 
For general (finite mixture) sources, the SLB is known to be tight only asymptotically as \(D\downarrow0\) \cite{Linkov1965,BiniaZakaiZiv1974,LinderZamir1994}.

We then characterize exactly when GRM meets the mixture SLB, and hence, the RDF for a GMM. For a
general finite mixture, this occurs at distortion \(D\) if and only if
every component RDF meets its own, unclipped SLB at the same distortion
\(D\). For a GMM, this yields
$
R^{ }(D)
=
R_{\mathrm{GRM}}(D)
$
throughout the fully active interval
\(D/d<\min_{k,m}\lambda_{k,m}\), where $\lambda_{k,m}$ denotes the $m$-th eigenvalue of the covariance matrix of the $k$-th component. Thus GRM is exactly Shannon-optimal on a
nonzero distortion interval at a finite source dimension. This conclusion does not extend from GMMs to arbitrary finite mixtures. This is because
the mixture SLB may be tight even though none of the component SLBs is
tight, in which case GRM does not attain the mixture SLB.
Note that GRM cannot operate at rates below the label sampling cost of $I(\mathbf{X};T)$, whereas $R(D)\rightarrow 0$ as $D \rightarrow D_{\max}$. For the low-rate region, 
we introduce stochastic
scalar quantization (SSQ), which simulates a reconstruction
index $J$ of cardinality at most \(M\), and lets the decoder output codeword $\mathbf{a}_J$ from a predesigned codebook. SSQ operates at the
alphabet-constrained information rate \(R_M^{ }\)
\cite{FinamorePearlman1980,PearlmanChekima1984}, which converges to $R(D)$, as $M \rightarrow \infty$. 
The deterministic counterpart of SSQ is the well-known Entropy-constrained VQ \cite{ChouLookabaughGray1989}, whose high-resolution behavior is governed by the
theory of Gish and Pierce \cite{GishPierce1968}, Gersho
\cite{Gersho1979}, and Zador \cite{Zador1982} (see also the survey
\cite{GrayNeuhoff1998}).

We analyze SSQ at the other end of the RDF, near $D_{\max}$, where we obtain the first-order expansion of $R(D)$, and show that a two-codeword codebook attains it.
We further study the limitations of SSQ, showing its suboptimality in the low-distortion region, where GRM is optimal, and
identifying a positive
nearest-neighbor distortion floor for every fixed \(M\).

The main contributions of this paper are as follows.

\begin{itemize}

\item We propose the Gibbs-relabeling method (GRM), which samples a posterior mixture label by channel simulation and then applies a component-conditional source code. The scheme attains the upper bound in \eqref{eq:intro} at label cost \(I(\mathbf{X};T)\) for every finite mixture.

\item We characterize exactly when GRM meets the mixture SLB, and hence the RDF. For a general finite mixture this occurs at distortion \(D\) if and only if every component RDF meets its own, unclipped SLB at the same \(D\). For a GMM the criterion yields
\(R^{ }(D)=R_{\mathrm{GRM}}(D)\) throughout the fully active interval
\(D/d<\min_{k,m}\lambda_{k,m}\).

\item We propose stochastic scalar quantization (SSQ) for the complementary low-rate region, as GRM cannot operate below \(I(\mathbf{X};T)\). SSQ samples a reconstruction index \(J\) of cardinality at most \(M\) by channel simulation, and the decoder emits a codeword \(\mathbf{a}_J\) from a finite codebook. Near \(D_{\max}\) we obtain the first-order expansion
\(R^{ }(D_{\max}-\delta)=\delta/(2\tau^2)+o(\delta)\)
and show that every SSQ codebook with \(M\ge 2\) attains the same coefficient.

\end{itemize}

The remainder of the paper is organized as follows.
Section~\ref{sec:problem} states the source coding problem and reviews channel simulation. Section~\ref{sec:grm} develops GRM, first for scalar GMMs
and then for general finite mixtures and vector GMMs, and establishes the
component-wise SLB criterion and its fully active GMM corollary.
Section~\ref{sec:lowrate} develops SSQ, the zero-rate expansion, and the
limitations imposed by a fixed reconstruction alphabet.
Section~\ref{sec:numerics} presents numerical illustrations, and
Section~\ref{sec:discussion} concludes. 

\color{black}
\section{Problem Formulation}
\label{sec:problem}

\subsection{Source and Distortion Measure}

Let \(\{\mathbf{X}_i\}_{i=1}^{\infty}\) be an independent and identically distributed (i.i.d.) source \(\mathbf{X}\in\mathbb{R}^d\), \(d\ge1\), where each element is distributed as a finite mixture
\begin{equation}
\label{eq:mix_gen}
    f_{\mathbf{X}}
    =
    \sum_{k=1}^{K}\pi_k f_k,
    \qquad
    P(T=k)=\pi_k,
    \qquad
    \mathbf{X}\mid T=k\sim f_k,
\end{equation}
with \(K<\infty\), \(\pi_k>0\), and \(\sum_{k=1}^{K}\pi_k=1\).
The posterior is
\begin{equation}
    P(T=k\mid\mathbf{X}=\mathbf{x})
    =
    \frac{
    \pi_k f_k(\mathbf{x})
    }{
    \sum_{i=1}^{K}
    \pi_i f_i(\mathbf{x})
    }.
\end{equation}
We define
\begin{equation}
    \boldsymbol{\mu}
    :=
    \mathbb{E}[\mathbf{X}]
    =
    \sum_{k=1}^{K}\pi_k\mathbb{E}[\mathbf{X}\mid T=k],
    \qquad
    \boldsymbol{\Sigma}_{\mathbf{X}}
    :=
    \operatorname{Cov}(\mathbf{X}),
\end{equation}
and
\begin{equation}
    D_{\max}
    :=
    \mathbb{E}\|\mathbf{X}-\boldsymbol{\mu}\|^2
    =
    \operatorname{tr}(\boldsymbol{\Sigma}_{\mathbf{X}}).
\end{equation}
The Gaussian special case of \eqref{eq:mix_gen} is the finite Gaussian mixture
\begin{equation}
\label{eq:gmm_density}
    f_{\mathbf{X}}(\mathbf{x})
    =
    \sum_{k=1}^{K}
    \pi_k\,
    \mathcal{N}(\mathbf{x};\boldsymbol{\mu}_k,\boldsymbol{\Sigma}_k),
\end{equation}
with \(\boldsymbol{\Sigma}_k\succ0\).
Write \(\lambda_{k,1},\ldots,\lambda_{k,d}\) for the eigenvalues of
\(\boldsymbol{\Sigma}_k\).
For a symmetric matrix \(\mathbf{A}\) we write
\(\lambda_{\min}(\mathbf{A})\) and \(\lambda_{\max}(\mathbf{A})\) for its
smallest and largest eigenvalues, so that
\(\min_{k,m}\lambda_{k,m}=\min_k\lambda_{\min}(\boldsymbol{\Sigma}_k)\).
Equivalently, there is a latent class \(T\in\{1,\ldots,K\}\) such that
\begin{equation}
\label{eq:gmm_gen}
    P(T=k)=\pi_k,
    \qquad
    \mathbf{X}\mid T=k
    \sim
    \mathcal{N}(\boldsymbol{\mu}_k,\boldsymbol{\Sigma}_k).
\end{equation}
Since each \(\boldsymbol{\Sigma}_k\succ0\), the component densities have full
support on \(\mathbb{R}^d\) and \(H(T\mid\mathbf{X})>0\), hence the class is
latent, rather than a deterministic function of \(\mathbf{X}\).
The posterior is
\begin{equation}
\label{eq:gmm_post}
    P(T=k\mid\mathbf{X}=\mathbf{x})
    =
    \frac{
    \pi_k\mathcal{N}(\mathbf{x};\boldsymbol{\mu}_k,\boldsymbol{\Sigma}_k)
    }{
    \sum_{i=1}^{K}
    \pi_i\mathcal{N}(\mathbf{x};\boldsymbol{\mu}_i,\boldsymbol{\Sigma}_i)
    }.
\end{equation}
We assume \(\mathbb{E}\|\mathbf{X}\|^2<\infty\), equivalently that each
\(f_k\) has finite second moment, so that \(D_{\max}<\infty\).
This holds automatically for GMMs.
Unless otherwise stated the source is the Gaussian mixture
\eqref{eq:gmm_density}--\eqref{eq:gmm_gen}.

The source and reproduction alphabets are
\(\mathcal{X}=\widehat{\mathcal{X}}=\mathbb{R}^d\).
Following Shannon and Berger
\cite{Shannon1959,Berger1971},
fidelity is measured by a single-letter nonnegative distortion
\(\delta:\mathcal{X}\times\widehat{\mathcal{X}}\to[0,\infty)\).
Throughout we consider squared error,
\begin{equation}
\label{eq:d_mse}
    \delta(\mathbf{x},\widehat{\mathbf{x}})
    =
    \|\mathbf{x}-\widehat{\mathbf{x}}\|^2,
\end{equation}
extended additively to \(n\)-blocks by
\begin{equation}
\label{eq:d_block}
    \delta_n(\mathbf{x}^n,\widehat{\mathbf{x}}^n)
    =
    \frac{1}{n}
    \sum_{i=1}^{n}
    \|\mathbf{x}_i-\widehat{\mathbf{x}}_i\|^2.
\end{equation}
Thus \(D_{\max}\) is the distortion incurred by reproducing each source sample by
the mean vector \(\boldsymbol{\mu}\).
All rates are in nats.
The scalar case \(d=1\) is written \(X\in\mathbb{R}\), with component
variances \(\sigma_k^2\) and rate-distortion function \(R(D)\).

\subsection{Block Codes and the Rate-Distortion Function}

Let \(\mathbf{X}^n=(\mathbf{X}_1,\ldots,\mathbf{X}_n)\) denote a source
block. An \((n,R)\) code is a pair of measurable maps
\begin{equation}
\label{eq:block_code}
    f_n:\mathcal{X}^n\to\bigl\{1,\ldots,\lceil e^{nR}\rceil\bigr\},
    \qquad
    g_n:\bigl\{1,\ldots,\lceil e^{nR}\rceil\bigr\}\to\widehat{\mathcal{X}}^n.
\end{equation}
The encoder describes \(\mathbf{X}^n\) by the index \(f_n(\mathbf{X}^n)\);
the decoder produces the reconstruction
\(\widehat{\mathbf{X}}^n:=g_n\bigl(f_n(\mathbf{X}^n)\bigr)\).
The expected distortion of the code is
\begin{equation}
    D_n
    :=
    \mathbb{E}\,
    \delta_n\bigl(\mathbf{X}^n,\widehat{\mathbf{X}}^n\bigr).
\end{equation}
A pair \((R,D)\) with \(R\ge0\) and \(D\ge0\) is \emph{achievable} if there
exists a sequence of \((n,R)\) codes with
\(\limsup_{n\to\infty} D_n\le D\).
The rate-distortion function is the infimum of rates at which distortion
\(D\) is achievable,
\begin{equation}
    R^{ }(D)
    :=
    \inf\bigl\{R\ge0:(R,D)\ \text{is achievable}\bigr\},
\end{equation}
with the convention that the infimum over the empty set is \(+\infty\).
The theorems below concern the infinite-blocklength regime
\(n\to\infty\); Section~\ref{sec:numerics} also reports variable-length
finite-blocklength results for channel simulation and entropy coding.

By Shannon's lossy source coding theorem
\cite{Shannon1959,Berger1971,CoverThomas},
for a memoryless source and a single-letter
distortion measure, the RDF can be characterized in the following single-letter form:
\begin{equation}
\label{eq:rdf}
    R^{ }(D)
    =
    \inf_{
    P_{\widehat{\mathbf{X}}\mid\mathbf{X}}:\,
    \mathbb{E}\|\mathbf{X}-\widehat{\mathbf{X}}\|^2\le D
    }
    I(\mathbf{X};\widehat{\mathbf{X}}).
\end{equation}
While squared error is unbounded, the identification still holds as
\(\mathbb{E}\|\mathbf{X}\|^2<\infty\), equivalently \(D_{\max}<\infty\).
In particular, \(R^{ }(D)=0\) for \(D\ge D_{\max}\).
Under squared error, \(R^{ }\) is convex, strictly decreasing on
\((0,D_{\max})\), and continuous
\cite{Berger1971,CoverThomas}.
It is therefore enough to exhibit test channels
\(P_{\widehat{\mathbf{X}}\mid\mathbf{X}}\) attaining \eqref{eq:rdf}.

Under squared error the SLB
\cite{Shannon1959,Berger1971} states that
\begin{equation}
\label{eq:slb}
    R^{ }(D)
    \ge
    \mathrm{SLB}(D)
    :=
    h(\mathbf{X})
    -
    \tfrac{d}{2}
    \log\!\bigl(2\pi e\,D/d\bigr).
\end{equation}
Equality holds at a given \(D\in(0,D_{\max})\) if and only if there exists a
reproduction \(\widehat{\mathbf{X}}\) of distortion \(D\) such that the
error \(\mathbf{X}-\widehat{\mathbf{X}}\) is independent of
\(\widehat{\mathbf{X}}\) and isotropic Gaussian
\cite{Berger1971,CoverThomas}.

\subsection{Channel Simulation}
\label{sec:rec}

The RDF \eqref{eq:rdf} is unchanged if the encoder
and decoder are permitted to share a random variable \(U\) independent of
the source: for every realization $U=u$ we have an ordinary $(n,R)$ code, so $R\geq R^{ }(D_u)$, and $R\geq \mathbb{E}[R^{ }(D_u)]\geq R^{ }(\mathbb{E}[D_u])$ by the convexity of the $R^{ }(D)$.
In the constructions below, we do not limit the common randomness, although in practice a common random seed is usually sufficient.

Let \(P_{J\mid\mathbf{X}}\) be a conditional distribution from $\mathbb{R}^d$ to a finite set $J$. A channel simulation code (also called a reverse channel code) for \(P_{J\mid\mathbf{X}}\) consists of common randomness \(U\) independent of
\(\mathbf{X}\) and maps
\begin{equation}
\label{eq:stoch_code}
    \mathrm{enc}_U:\mathcal{X}\to\{0,1\}^\ast,
    \qquad
    \mathrm{dec}_U:\{0,1\}^\ast\to\mathcal{J}
\end{equation}
that are deterministic given \(U\) and such that for each $u$ the set $\{ \mathrm{enc}_u(\mathbf{x})\}$ is prefix-free and
\begin{equation}
\label{eq:stoch_sim}
    \mathrm{dec}_U\bigl(\mathrm{enc}_U(\mathbf{x})\bigr)
    \sim
    P_{J\mid\mathbf{X}=\mathbf{x}}
\end{equation}
for almost every \(\mathbf{x}\).
The simulation is exact; that is, the induced joint law of $(\mathbf{X},J)$ is $P_\mathbf{X} \cdot P_{J|\mathbf{X}}$. 
Write \(L\) for the length of
\(\mathrm{enc}_U(\mathbf{X})\), measured in nats.
Such codes exist for every discrete kernel, with \cite{LiElGamal2018,FlamichGunduzREC}
\begin{equation}
\label{eq:rec_oneshot}
    I(\mathbf{X};J)
    \le
    \mathbb{E}[L]
    \le
    I(\mathbf{X};J)
    +
    O\bigl(\log\bigl(I(\mathbf{X};J)+1\bigr)\bigr) \mathrm{nats}.
\end{equation}
Applied to the product kernel
\(P_{J^n\mid\mathbf{X}^n}=\prod_{i=1}^{n}P_{J_i\mid\mathbf{X}_i}\)
on an i.i.d.\ block, \eqref{eq:rec_oneshot} gives expected length
\(n I(\mathbf{X};J)+O\bigl(\log\bigl(n I(\mathbf{X};J)+1\bigr)\bigr)\),
hence \(\frac{1}{n}\mathbb{E}[L_n] \rightarrow I(\mathbf{X};J)\).
Thus every discrete kernel \(P_{J\mid\mathbf{X}}\) with finite
\(I(\mathbf{X};J)\) is exactly realizable at that per-letter rate.

\section{Gibbs Relabeling for Finite Mixtures and GMMs}
\label{sec:grm}

We first develop GRM for scalar GMMs and establish its high-rate
optimality. We then extend the construction to general finite mixtures
on \(\mathbb{R}^d\), characterize when it meets the SLB, and specialize
the resulting criterion to multivariate GMMs.
\color{black}

\subsection{Scalar GMMs}
\label{sec:grm_scalar}

Specializing \eqref{eq:gmm_density} to \(d=1\) gives
\begin{equation}
    f_X(x)
    =
    \sum_{k=1}^{K}
    \pi_k\mathcal{N}(x;\mu_k,\sigma_k^2),
\end{equation}
where \(\sigma_k^2>0\). Equivalently, one first draws
\(T\sim\mathrm{Categorical}(\pi_1,\ldots,\pi_K)\) and then draws
\begin{equation}
    X\mid T=k\sim\mathcal{N}(\mu_k,\sigma_k^2).
\end{equation}

This representation suggests reversing the generative procedure for
source coding: first coordinate a label drawn from the posterior
\(P_{T\mid X}\), and then encode \(X\) under the corresponding component
law. Throughout this section, the encoder and decoder know the GMM
parameters and have access to the common randomness specified in
Section~\ref{sec:rec}.

\begin{algorithm}[t]
\DontPrintSemicolon
\SetAlgoLined
\caption{Gibbs-relabeling method (scalar GMM)}
\label{alg:grm_scalar}
\KwIn{parameters \(\{(\pi_k,\mu_k,\sigma_k^2)\}_{k=1}^{K}\); distortion \(D>0\); common randomness \(U\); realization \(x^n\in\mathbb{R}^n\)}
\KwOut{reconstruction \(\widehat{x}^n\in\mathbb{R}^n\)}
Determine, independently of \(x^n\), the allocation \(\{D_k^{\star}\}_{k=1}^{K}\) by reverse water-filling:
\(D_k^{\star}=\min(\sigma_k^2,\theta)\) with \(\theta>0\) such that \(\sum_{k=1}^{K}\pi_k D_k^{\star}=D\)\;

For \(i=1,\ldots,n\), 
\qquad \qquad \qquad \qquad \qquad \qquad \qquad \qquad
set
\(
p_k(x_i)
\gets
\frac{\pi_k\,\mathcal{N}(x_i;\mu_k,\sigma_k^2)}
{\sum_{j=1}^{K}\pi_j\,\mathcal{N}(x_i;\mu_j,\sigma_j^2)}
\)
for \(k=1,\ldots,K\)\;

Using \(U\), channel simulate
\(\widetilde{T}^n\sim\prod_{i=1}^{n}(p_1(x_i),\ldots,p_K(x_i))\)\;

For each \(k\), apply any RDF-achieving code for
\(\mathcal{N}(\mu_k,\sigma_k^2)\) at distortion \(D_k^{\star}\)
to the subsequence \(\{x_i:\widetilde{T}_i=k\}\)
(an explicit test channel is recorded in Appendix~\ref{app:gauss_tc}),
yielding an index \(W_k\)\;

From \((\widetilde{T}^n,\{W_k\}_{k=1}^{K})\), output the reconstruction \(\widehat{x}^n\) of those codes\;

\end{algorithm}

We begin by verifying the distribution induced by posterior relabeling.

\begin{lemma}[Joint-law preservation]
\label{lem:joint_scalar}
The pair \((\widetilde{T}^n,X^n)\) generated by Algorithm~\ref{alg:grm_scalar}
has the same law as the generative pair \((T^n,X^n)\). In particular,
for every \(i\in\{1,\ldots,n\}\),
    \begin{equation}
\label{eq:joint_scalar}
    X_i\mid(\widetilde{T}_i=k)
    \sim
    \mathcal{N}(\mu_k,\sigma_k^2),
    \qquad
    I(X^n;\widetilde{T}^n)
    =
    n\,I(X;T).
    \end{equation}
\end{lemma}
\begin{proof}
Both constructions use the same source letter \(X^n\sim f_X^{\otimes n}\). It remains to
match the kernels \(P_{\widetilde{T}^n\mid X^n}\) and \(P_{T^n\mid X^n}\); the
joints then factor in the same way.

Since each \(\sigma_k^2>0\), one has \(f_X(x)>0\) for all \(x\), so the
weights in Algorithm~\ref{alg:grm_scalar} are well-defined. Bayes' rule
for the generative pair \((T,X)\) gives
\begin{equation}
\label{eq:pk_bayes}
    p_k(x)
    =
    \frac{
    \pi_k\,\mathcal{N}(x;\mu_k,\sigma_k^2)
    }{
    f_X(x)
    }
    =
    P(T=k\mid X=x).
\end{equation}
The simulation step is exact \eqref{eq:stoch_sim}, so
\(x\mapsto\bigl(p_1(x),\ldots,p_K(x)\bigr)\)
\cite{LiElGamal2018,FlamichGunduzREC}.
Thus \(P_{\widetilde{T}\mid X}=P_{T\mid X}\), and therefore
\begin{equation}
\label{eq:joint_factor}
    P_{\widetilde{T},X}
    =
    P_{\widetilde{T}\mid X}\,P_X
    =
    P_{T\mid X}\,P_X
    =
    P_{T,X}.
\end{equation}
In particular, \(P(\widetilde{T}=k)=\pi_k>0\) and
\(X\mid(\widetilde{T}=k)\sim X\mid(T=k)\sim\mathcal{N}(\mu_k,\sigma_k^2)\).
Mutual information depends on a pair only through its joint law, so
\(I(X;\widetilde{T})=I(X;T)\).
\end{proof}

Once \(\widetilde{T}\) is available to both the encoder and decoder,
Lemma~\ref{lem:joint_scalar} reduces the remaining step to conditional
rate-distortion coding of a Gaussian mixture with an observed component
label. Define
\begin{equation}
\label{eq:Rcond_XT}
    R_{X\mid\widetilde{T}}(D)
    :=
    \inf_{
    P_{\widehat{X}\mid X,\widetilde{T}}:\,
    \mathbb{E}(X-\widehat{X})^2\le D
    }
    I(X;\widehat{X}\mid\widetilde{T}).
\end{equation}
Since \(\widetilde{T}\) is discrete,
\(I(X;\widehat{X}\mid\widetilde{T})=\sum_{k}\pi_k I(X;\widehat{X}\mid\widetilde{T}=k)\).
Writing \(D_k:=\mathbb{E}\bigl[(X-\widehat{X})^2\mid\widetilde{T}=k\bigr]\),
one has \(\sum_{k}\pi_k D_k=D\), and each inner minimization is the MSE RDF
of \(\mathcal{N}(\mu_k,\sigma_k^2)\),
\begin{equation}
\label{eq:Rk_scalar}
    R_k(D_k)
    :=
    \begin{cases}
    \tfrac12\log(\sigma_k^2/D_k), & 0<D_k\le\sigma_k^2,\\
    0, & D_k\ge\sigma_k^2.
    \end{cases}
\end{equation}
Thus, \eqref{eq:Rcond_XT} is exactly the weighted allocation
\begin{equation}
\label{eq:Rcond_scalar}
    R_{X\mid\widetilde{T}}(D)
    =
    \inf
    \Biggl\{
    \sum_{k=1}^{K}
    \pi_k R_k(D_k)
    :
    D_k\ge0,\
    \sum_{k=1}^{K}\pi_k D_k=D
    \Biggr\}.
\end{equation}
The unique minimizer of \eqref{eq:Rcond_scalar} is reverse water-filling
on the parallel Gaussians \(\{X\mid\widetilde{T}=k\}_{k=1}^{K}\)
\cite{Berger1971,CoverThomas}:
\begin{equation}
\label{eq:RWF_scalar}
    D_k^{\star}
    =
    \min(\sigma_k^2,\theta),
\end{equation}
where the water level \(\theta>0\) is the unique solution of
\begin{equation}
\label{eq:theta_scalar}
    \sum_{k=1}^{K}
    \pi_k\min(\sigma_k^2,\theta)
    =
    D.
\end{equation}
A self-contained Karush-Kuhn-Tucker (KKT) derivation is given in
Appendix~\ref{app:rwf_scalar}.
Substituting \eqref{eq:RWF_scalar} into \eqref{eq:Rcond_scalar} yields
\begin{equation}
\label{eq:Rcond_closed}
    R_{X\mid\widetilde{T}}(D)
    =
    \frac12
    \sum_{k:\,\sigma_k^2>\theta}
    \pi_k\log(\sigma_k^2/\theta).
\end{equation}
By the reverse Shannon theorem (Section \ref{sec:rec}), simulating the \(n\)-letter product posterior requires
\(I(X;\widetilde{T})+o(1)\) nats per letter as \(n\to\infty\)
\cite{LiElGamal2018,FlamichGunduzREC}. By
Lemma~\ref{lem:joint_scalar}, this mutual information equals \(I(X;T)\).
Adding the conditional rate
\eqref{eq:Rcond_closed} therefore gives, for every
\(D\in\bigl(0,\sum_{k=1}^{K}\pi_k\sigma_k^2\bigr]\),
\begin{equation}
\label{eq:R_GRM_scalar}
    R_{\mathrm{GRM}}(D)
    :=
    I(X;T)
    +
    R_{X\mid\widetilde{T}}(D).
\end{equation}
This is the rate of a feasible test channel through Algorithm~\ref{alg:grm_scalar},
so \(R(D)\le R_{\mathrm{GRM}}(D)\) by definition of the RDF.
Together with the Shannon lower bound \eqref{eq:slb}, one has the sandwich
\(\mathrm{SLB}(D)\le R(D)\le R_{\mathrm{GRM}}(D)\).
The following theorem identifies the distortions at which these inequalities hold with equality.

\begin{theorem}[High-rate optimality of GRM for scalar GMMs]
\label{thm:grm_scalar}
If \(D<\min_k\sigma_k^2\), then
\begin{equation}
    R_{\mathrm{GRM}}(D)
    =
    h(X)-\tfrac12\log(2\pi e\,D)
    =
    R(D).
\end{equation}
\end{theorem}

\begin{proof}
Because \(D<\min_k\sigma_k^2\), the choice \(\theta=D\) satisfies
\eqref{eq:theta_scalar} and lies strictly below every component variance.
Uniqueness of the water level (Appendix~\ref{app:rwf_scalar}) therefore
yields \(D_k^{\star}=D\) for every \(k\). Substituting into \eqref{eq:R_GRM_scalar} and
\eqref{eq:Rcond_closed} gives
\begin{equation}
    R_{\mathrm{GRM}}(D)
    =
    I(X;T)
    +
    \sum_{k=1}^{K}
    \pi_k\cdot\tfrac12\log(\sigma_k^2/D).
\end{equation}
The differential entropy of \(\mathcal{N}(\mu_k,\sigma_k^2)\) is
\(\tfrac12\log(2\pi e\sigma_k^2)\), hence
\begin{align}
    &\sum_{k=1}^{K} \pi_k\cdot\tfrac12\log(\sigma_k^2/D)\nonumber\\
    =&
    \sum_{k=1}^{K}
    \pi_k
    \Bigl(
    \tfrac12\log(2\pi e\sigma_k^2)
    -
    \tfrac12\log(2\pi e\,D)
    \Bigr)
    \nonumber\\
    =&
    h(X\mid T)
    -
    \tfrac12\log(2\pi e\,D).
\end{align}
Combining with \(I(X;T)+h(X\mid T)=h(X)\) yields
\begin{align}
    R_{\mathrm{GRM}}(D)
    &=
    I(X;T)
    +
    h(X\mid T)
    -
    \tfrac12\log(2\pi e\,D)
    \nonumber\\
    &=
    h(X)
    -
    \tfrac12\log(2\pi e\,D)
    =
    \mathrm{SLB}(D).
\end{align}
Together with \(\mathrm{SLB}(D)\leq R(D)\le R_{\mathrm{GRM}}(D)\), one has
\(R(D)=\mathrm{SLB}(D)=R_{\mathrm{GRM}}(D)\).
\end{proof}

\begin{remark}[Relation between tightness of SLB and condition of Theorem \ref{thm:grm_scalar}]
\label{re:gmm_slb_coincidence}
The condition in Theorem~\ref{thm:grm_scalar} coincides
with the condition for the tightness of the mixture SLB.
The coincidence is due to a property of Gaussians, which persists for
multivariate GMMs but fails for general finite
mixtures, where the mixture SLB may be tight while GRM is not.
We further derive this in
Section~\ref{sec:grm_general}.
\end{remark}

\begin{remark}[Converse of Theorem \ref{thm:grm_scalar}]
\label{rem:converse}
We note that the converse of Theorem \ref{thm:grm_scalar} is also true. It follows by contradiction.
Suppose GRM attains \(R(D)\) at some distortion $D$ with
\(\min_k\sigma_k^2<D<D_{\max}\).
In the region, for the mixture \(R(D)>\mathrm{SLB}(D)\), so
every RDF-achieving reproduction is purely discrete
\cite[Theorem 3]{Rose1994}.
Reverse water-filling, however, still leaves at least one component
active.
GRM encodes that component by the scalar Gaussian test channel at a
distortion strictly below its variance, where \(R(D)=\mathrm{SLB}(D)\) for a
Gaussian source, and therefore, yields an absolutely continuous part in the
conditional reconstruction \cite[Theorem 3]{Rose1994}.
The resulting unconditional law of \(\widehat{X}\) then has a density on
\(\mathbb{R}\), contradicting discreteness.
Hence GRM cannot attain \(R(D)\) for \(D>\min_k\sigma_k^2\).
Combined with Theorem~\ref{thm:grm_scalar} and continuity at the
threshold, GRM is Shannon-optimal if and only if
\(D\le\min_k\sigma_k^2\).
However this converse does not extend to multivariate cases naturally, as scalar Gaussian is the unique source distribution that admits continuous reconstruction for the entire range of positive rates \cite[Section IV-C]{Rose1994}. The only exception is where all the components are isotropic multivariate Gaussians, which is equivalent under
additive squared error to \(d\) i.i.d.\ scalar Gaussians of variance
\(\sigma^2\).
\end{remark}
\color{black}
\subsection{Extension to Multivariate Finite Mixtures}
\label{sec:grm_method}

The derivation for a general finite mixture on \(\mathbb{R}^d\) parallels
Section~\ref{sec:grm_scalar} through the definition of \(R_{\mathrm{GRM}}\).
We note that Gaussian reverse water-filling enters only as the explicit form of the
component RDFs and is not used for joint-law preservation or the separable
allocation. We briefly provide the extension below.

Let \(\mathbf{X}\in\mathbb{R}^d\) be a general finite mixture source
\eqref{eq:mix_gen}, and write \(R_k(D)\) for the MSE RDF of \(f_k\).
Given \(\mathbf{X}=\mathbf{x}\), GRM samples
\(\widetilde{T}\sim P_{T\mid\mathbf{X}=\mathbf{x}}\), communicates
\(\widetilde{T}\) by channel simulation, and applies an RDF-achieving code for
\(f_{\widetilde{T}}\) subject to a shared distortion budget.
The same argument as in Lemma~\ref{lem:joint_scalar} yields
\((\widetilde{T},\mathbf{X})\stackrel{d}{=}(T,\mathbf{X})\), hence
\(\mathbf{X}\mid(\widetilde{T}=k)\sim f_k\) and
\(I(\mathbf{X};\widetilde{T})=I(\mathbf{X};T)\).
As \(T\) is discrete and \(\mathbf{X}\mid T=k\sim f_k\), the
conditional RDF given the label is exactly the weighted allocation
\begin{equation}
\label{eq:Rcond_sep}
    R_{X|T}(D)
    :=
    \inf_{
    D_k\ge0,\,
    \sum_{k=1}^{K}\pi_k D_k=D
    }
    \sum_{k=1}^{K}
    \pi_k R_k(D_k),
\end{equation}
where the superscript marks the allocation over component distortions.
Adding the simulation rate \(I(\mathbf{X};T)\)
\cite{LiElGamal2018,FlamichGunduzREC} gives, for every \(D\) at which the
constraint set in \eqref{eq:Rcond_sep} is nonempty,
\begin{equation}
\label{eq:R_GRM}
    R_{\mathrm{GRM}}(D)
    :=
    I(\mathbf{X};T)
    +
    R_{\mathbf{X}|T}(D).
\end{equation}
This is the rate of a feasible test channel, so
\(R^{ }(D)\le R_{\mathrm{GRM}}(D)\).
The sandwich \(\mathrm{SLB}(D)\le R^{ }(D)\le R_{\mathrm{GRM}}(D)\)
therefore persists for every finite mixture.
One recovers the scalar GMM of Section~\ref{sec:grm_scalar} by letting \(d=1\),
\(f_k=\mathcal{N}(\mu_k,\sigma_k^2)\), in which \eqref{eq:R_GRM} reduces to
\eqref{eq:R_GRM_scalar}.

\subsection{SLB Tightness and a Necessary and Sufficient Criterion}
\label{sec:grm_general}

Recall \(\mathrm{SLB}(D)\) from \eqref{eq:slb}, and write
\begin{equation}
    \mathrm{SLB}_k(D)
    :=
    h(\mathbf{X}|T=k)
    -
    \frac{d}{2}
    \log\!\bigl(2\pi e\,D/d\bigr)
\end{equation}
for the corresponding bound on component \(k\).
To address Remark~\ref{re:gmm_slb_coincidence}, we start this section by
pointing out that although tightness of \(\mathrm{SLB}(\cdot)\) is equivalent to
tightness of every \(\mathrm{SLB}_k(\cdot)\) for GMMs, this property does not hold
for general finite mixtures, as the next example shows. Let \(\varphi_D\) be the density of
\(\mathcal{N}(0,D)\), and let \(\varphi:=\varphi_1\).

\begin{lemma}[Decomposition of a Gaussian density]
\label{lem:ex_slb_construction}
There exist distinct, everywhere positive probability densities
\(f_1,f_2\) on \(\mathbb{R}\), each with mean zero, variance one and finite differential entropy, i.e., $|h(f_k)|<\infty$, such
that
$
    (f_1+f_2)/2=\varphi.
$
Furthermore, for each \(k\in\{1,2\}\), there do not exist \(D>0\) and a
probability measure \(\nu\) such that
\[
    f_k(x)
    =
    (\nu*\varphi_D)(x)
    :=
    \int_{\mathbb{R}}\varphi_D(x-y)\,\nu(dy),
    \qquad x\in\mathbb{R}.
\]
\end{lemma}
\begin{proof}
See Appendix~\ref{app:ex_slb}.
\end{proof}

\begin{example}[Mixture SLB tightness without component-wise tightness]
\label{ex:slb_not_grm}
Let \(P(T=1)=P(T=2)=1/2\), and let \(X\) conditioned on \(T=k\) have
density \(f_k\) from Lemma~\ref{lem:ex_slb_construction}. The marginal
density of \(X\) is \((f_1+f_2)/2=\varphi\). Thus \(X\) is standard
Gaussian and
\begin{equation}
    R(D)=\mathrm{SLB}(D)=\tfrac12\log(1/D),
    \qquad D\in(0,1].
\end{equation}
Since $h(f_k)$ is finite, $\mathrm{SLB}_k(D)$ is finite and the equality criterion of the SLB applies. For a scalar source under squared error, equality in the SLB at
\(D>0\) holds if and only if the source density can be written as
\(\nu*\varphi_D\) for some probability measure \(\nu\)
\cite{Berger1971,CoverThomas}. Lemma~\ref{lem:ex_slb_construction}
rules out this representation for both components. Therefore,
\[
    R_k(D)>\mathrm{SLB}_k(D),
    \qquad k\in\{1,2\},\quad D\in(0,1).
\]
The mixture SLB is therefore tight throughout \((0,1)\), whereas both
component SLB inequalities are strict throughout that interval.
\end{example}

Example~\ref{ex:slb_not_grm} shows that mixture SLB tightness alone does
not imply component-wise tightness. The next theorem shows that
component-wise tightness is necessary and sufficient for GRM to meet the
mixture SLB. Under this condition, the sandwich
\(\mathrm{SLB}(D)\le R^{ }(D)\le R_{\mathrm{GRM}}(D)\) collapses and
GRM is Shannon-optimal.

\color{black}
\begin{theorem}[Necessary and sufficient conditions for GRM to meet the SLB]
\label{prop:comp_slb}
Fix \(D>0\).
The following are equivalent:
\begin{enumerate}
    \item \(R_k(D)=\mathrm{SLB}_k(D)\) for every \(k\);
    \item \(R_{\mathrm{GRM}}(D)=\mathrm{SLB}(D)\).
\end{enumerate}
In either case 
\(R_{\mathrm{GRM}}(D)=\mathrm{SLB}(D)=R^{ }(D)\).
\end{theorem}

\begin{proof}
\textit{Sufficiency.}
Assume \(R_k(D)=\mathrm{SLB}_k(D)\) for every \(k\).
The choice \(D_k=D\) is admissible in \eqref{eq:Rcond_sep}, so
\begin{align}
    R_{\mathbf{X}|T}(D)
    &\le
    \sum_{k=1}^{K}
    \pi_k R_k(D)
    =
    \sum_{k=1}^{K}
    \pi_k\mathrm{SLB}_k(D)
    \nonumber\\
    &=
    h(\mathbf{X}|T)
    -
    \frac{d}{2}
    \log\!\bigl(2\pi e\,D/d\bigr).
\end{align}
Adding the label cost and applying definition of mutual information yields
\(R_{\mathrm{GRM}}(D)\le\mathrm{SLB}(D)\).
The sandwich
\(\mathrm{SLB}(D)\le R^{ }(D)\le R_{\mathrm{GRM}}(D)\)
forces
\(R_{\mathrm{GRM}}(D)=\mathrm{SLB}(D)=R^{ }(D)\).

\textit{Necessity.}
Assume \(R_{\mathrm{GRM}}(D)=\mathrm{SLB}(D)\).
For every admissible allocation \(\{D_k\}\) with
\(\sum_k\pi_k D_k=D\),
\begin{align}
    \sum_{k=1}^{K}
    \pi_k R_k(D_k)
    &\ge
    \sum_{k=1}^{K}
    \pi_k\mathrm{SLB}_k(D_k)
    \nonumber\\
    &=
    h(\mathbf{X}|T)
    -
    \frac{d}{2}
    \log(2\pi e) 
    \nonumber\\
    & \qquad \qquad \, \, \, \, \, -
    \frac{d}{2}
    \sum_{k=1}^{K}
    \pi_k\log(D_k/d).
\end{align}
Concavity of \(\log\) gives
\(\sum_k\pi_k\log D_k\le\log D\), with equality if and only if
\(D_k=D\) for every \(k\) since $\pi_k>0$.
Hence
\begin{equation}
    \sum_{k=1}^{K}
    \pi_k R_k(D_k)
    \ge
    \sum_{k=1}^{K}
    \pi_k\mathrm{SLB}_k(D),
\end{equation}
and therefore
\(R_{X|T}^{ }(D)\ge\sum_k\pi_k\mathrm{SLB}_k(D)\).
But
\(R_{\mathrm{GRM}}(D)=\mathrm{SLB}(D)\) rearranges as
\(R_{X|T}^{ }(D)=\sum_k\pi_k\mathrm{SLB}_k(D)\).
The infimum defining \(R_{X|T}^{ }(D)\) is thus attained if and only if both equalities hold, which forces
\(D_k\equiv D\) and \(R_k(D)=\mathrm{SLB}_k(D)\) for every \(k\).
\end{proof}

\begin{remark}[Convention of tightness of SLB]
    We note that condition~(1) requires \emph{exact} tightness against the
unclipped Shannon formula \(\mathrm{SLB}_k(D)\).
Here we lift the convention of identifying a rate-distortion function with
\(\max\{0,\mathrm{SLB}_k(D)\}\). Specifically, whenever
\(\mathrm{SLB}_k(D)<0\), one has \(R_k(D)\ge0>\mathrm{SLB}_k(D)\), so~(1)
fails even though the clipped bound is met.
Under finite differential entropy, each map
\(D\mapsto\mathrm{SLB}_k(D)\) is continuous and strictly decreasing from
\(+\infty\) to \(-\infty\).
Hence, whenever \(R_k(D)=0\) for some component $k$, \(\mathrm{SLB}_k(D)<0\) almost surely and the condition 1) fails.
\end{remark}

Hence for general finite mixtures, component-wise SLB tightness is sufficient for GRM to be optimal. When some \(R_k(D)>\mathrm{SLB}_k(D)\), Theorem~\ref{prop:comp_slb} forces
\(R_{\mathrm{GRM}}(D)>\mathrm{SLB}(D)\).
This does not preclude
\(R_{\mathrm{GRM}}(D)=R^{ }(D)>\mathrm{SLB}(D)\).
However, Fig.~\ref{fig:bimodal_2d} gives an example in which GRM
numerically separates from the RDF estimate once component-wise SLB
tightness fails.
The comparison used in the proof of Theorem~\ref{prop:comp_slb}
nevertheless yields an asymptotic statement as \(D\downarrow0\).

\begin{corollary}[Asymptotic high-rate optimality of GRM]
\label{cor:grm_asymp}
Assume \(h(\mathbf{X})>-\infty\) and
\begin{equation}
\label{eq:comp_slb_asymp}
    \lim_{D\downarrow 0}
    \bigl(R_k(D)-\mathrm{SLB}_k(D)\bigr)
    =
    0,
    \qquad
    k=1,\ldots,K.
\end{equation}
Then
\begin{equation}
\label{eq:grm_asymp}
    \lim_{D\downarrow 0}
    \bigl(R_{\mathrm{GRM}}(D)-\mathrm{SLB}(D)\bigr)
    =
    \lim_{D\downarrow 0}
    \bigl(R_{\mathrm{GRM}}(D)-R^{ }(D)\bigr)
    =
    0.
\end{equation}
\end{corollary}

\begin{proof}
The identity
\(\mathrm{SLB}(D)=I(\mathbf{X};T)+\sum_{k=1}^{K}\pi_k\mathrm{SLB}_k(D)\)
follows from \(h(\mathbf{X})=I(\mathbf{X};T)+h(\mathbf{X}|T)\).
For every \(D>0\) the equal allocation \(D_k=D\) is admissible in
\eqref{eq:Rcond_sep}, so
\(R_{\mathbf{X}|T}^{ }(D)\le\sum_{k=1}^{K}\pi_k R_k(D)\) and
\begin{equation}
    0
    \le
    R_{\mathrm{GRM}}(D)-\mathrm{SLB}(D)
    \le
    \sum_{k=1}^{K}
    \pi_k\bigl(R_k(D)-\mathrm{SLB}_k(D)\bigr).
\end{equation}
The rightmost expression tends to zero as \(D\downarrow0\) by
\eqref{eq:comp_slb_asymp}.
The remaining claim follows from
\(\mathrm{SLB}(D)\le R^{ }(D)\le R_{\mathrm{GRM}}(D)\).
\end{proof}

The condition \eqref{eq:comp_slb_asymp} holds whenever each component
density \(f_k\) satisfies \(|h(\mathbf{X}\mid T=k)|<\infty\) and
\(H(\lfloor\mathbf{X}\rfloor\mid T=k)<\infty\)
\cite[Theorem~2]{Koch2016SLB}.
We note that for Gaussian components, \(R_k=\mathrm{SLB}_k\)
holds on a positive-length interval, instead of only in the limit \(D\downarrow0\).
We study this interval in \(\mathbb{R}^d\) in the next subsection.

\subsection{Multivariate Gaussian Mixtures}
\label{sec:grm_gmm}

Let each \(f_k=\mathcal{N}(\boldsymbol{\mu}_k,\boldsymbol{\Sigma}_k)\) with
\(\boldsymbol{\Sigma}_k
=
\mathbf{U}_k\operatorname{diag}(\lambda_{k,1},\ldots,\lambda_{k,d})
\mathbf{U}_k^{\mathsf{T}}\)
and \(\lambda_{k,m}>0\) for $m=1,...,d$.
The MSE RDF of a nondegenerate Gaussian is obtained through reverse water-filling on its
eigenvalues \cite{Berger1971,CoverThomas}, so
Theorem~\ref{prop:comp_slb} specializes to an explicit eigen-interval.
For completeness, the resulting GRM rate is
\begin{equation}
\label{eq:R_GRM_vec}
    R_{\mathrm{GRM}}(D)
    =
    I(\mathbf{X};T)
    +
    \frac 1 2
    \sum_{k=1}^{K}
    \pi_k
    \sum_{m:\,\lambda_{k,m}>\theta}
    \log\frac{\lambda_{k,m}}{\theta},
\end{equation}
where \(\theta>0\) is the unique common water level satisfying
\begin{equation}
\label{eq:theta_vec}
    D
    =
    \sum_{k=1}^{K}
    \pi_k
    \sum_{m=1}^{d}
    \min(\lambda_{k,m},\theta),
\end{equation}
and \(D_{k,m}^{\star}=\min(\lambda_{k,m},\theta)\)
(Appendix~\ref{app:rwf_vec}).

The optimality of GRM is then straightforward as a corollary of Theorem~\ref{prop:comp_slb}.

\begin{corollary}[High-rate optimality of GRM for GMMs]
\label{thm:grm}
If
\begin{equation}
    \frac{D}{d}
    <
    \min_{k,m}\lambda_{k,m},
\end{equation}
then
\begin{equation}
    R_{\mathrm{GRM}}(D)
    =
    h(\mathbf{X})
    -
    \frac{d}{2}
    \log\!\bigl(2\pi e\,D/d\bigr)
    =
    R^{ }(D).
\end{equation}
In particular, GRM attains the Shannon RDF throughout this distortion
regime.
\end{corollary}

\begin{proof}
Each nondegenerate Gaussian in \(\mathbb{R}^d\) meets its MSE SLB at total
distortion \(D\) if and only if reverse water-filling is fully active, i.e.,
if and only if \(D/d\) does not exceed the smallest eigenvalue of its
covariance \cite{Berger1971,CoverThomas}.
Under the stated hypothesis,
\(D/d<\lambda_{\min}(\boldsymbol{\Sigma}_k)\) for every \(k\), so
\(R_k(D)=\mathrm{SLB}_k(D)\) for every \(k\).
Theorem~\ref{prop:comp_slb} therefore yields
\(R_{\mathrm{GRM}}(D)=\mathrm{SLB}(D)=R^{ }(D)\).
The identification \(\theta=D/d\) is the unique solution of
\eqref{eq:theta_vec} in this range
(see Appendix~\ref{app:rwf_vec}).
\end{proof}

\begin{remark}[Boundary and inactive directions]
The equality of Corollary~\ref{thm:grm} extends to
\(D/d=\min_{k,m}\lambda_{k,m}\) by continuity of reverse water-filling.
Beyond this threshold some eigen-directions are inactive: they spend zero
rate and the reconstruction error is no longer a full-rank isotropic
Gaussian of variance \(D/d\) per dimension.
The SLB, which assumes exactly that error law, therefore cannot match
\(R_{\mathrm{GRM}}(D)\).
Achievability of \(R_{\mathrm{GRM}}(D)\) remains, while whether
\(R_{\mathrm{GRM}}(D)=R^{ }(D)\) continues to hold is not settled.
In particular, for \(d=1\) there is a single eigenvalue per component and
the fully active regime is exactly \(D\le\min_k\sigma_k^2\), recovering Theorem \ref{thm:grm_scalar}.
\end{remark}

\begin{remark}[Successive refinability of GRM for GMM]
    We note that for GMM source GRM is naturally a successively refinable scheme. This follows since once $\widetilde{T}$ is communicated, $\mathbf{X}|\widetilde{T}\sim \mathcal{N}$, and hence is successively refinable in MSE \cite{75242}.
\end{remark}
\section{Stochastic Scalar Quantization and the Zero-Rate Boundary}
\label{sec:lowrate}

In the previous section, we constructed GRM and identified a high-rate
interval on which it attains \(R^{ }\).
The construction however always simulates a posterior sample of \(T\), setting the label cost $I(\mathbf{X};T)$ as a lower bound independent of $D$, whereas
\(R^{ }(D)\to0\) as \(D\to D_{\max}\).
We now turn to the complementary low-rate regime and propose SSQ.
We retain a discrete reconstruction index, but unlike in GRM, it does not depend on the mixture
posterior or the support set. For each index the reconstruction is a fixed codeword from an
arbitrary finite codebook.

Throughout this section, \(\mathbf{X}\in\mathbb{R}^d\) is a finite GMM.

\subsection{Stochastic Scalar Quantization (SSQ)}
\label{sec:fa_codes}

We construct SSQ as follows. Let \(M\ge1\) be an integer, let \(\mathcal{J}\) be a finite alphabet with $|\mathcal{J}|\leq M$,
and let \(J\) be a \(\mathcal{J}\)-valued random variable on the same
probability space as \(\mathbf{X}\), specified by a conditional kernel
\(P_{J\mid\mathbf{X}}\).
A codebook is a deterministic map \(j\mapsto\mathbf{a}_j\in\mathbb{R}^d\),
\(j\in\mathcal{J}\), and the associated reconstruction is
\(\widehat{\mathbf{X}}:=\mathbf{a}_J\).

Encoder and decoder share the codebook.
Given \(\mathbf{X}=\mathbf{x}\), the encoder and decoder simulate
\(J\sim P_{J\mid\mathbf{X}=\mathbf{x}}\) (Section \ref{sec:rec}), and the decoder emits
\(\mathbf{a}_J\).
Applied to the \(n\)-letter product kernel, the simulation rate is \(I(\mathbf{X};J)\) nats per
letter as \(n\to\infty\); the expected distortion is
\(\mathbb{E}\|\mathbf{X}-\mathbf{a}_J\|^2\).
Note that SSQ is a stochastic quantization scheme that is, for a given $\mathbf{X}=\mathbf{x}$ the reconstruction $\widehat{\mathbf{X}}$ is a random variable governed by the common randomness, while scalar quantization (SQ) is the deterministic special case in which each
letter is mapped independently by a nearest-neighbor partition of
\(\mathbb{R}^d\) onto the codebook, and in this case
\(I(\mathbf{X};J)=H(J)\).
The SSQ information rate is the infimal such rate
subject to a distortion budget \(D\):
\begin{equation}
\label{eq:rdf_M}
    R_M^{ }(D)
    :=
    \inf_{
    \substack{
    P_{J\mid\mathbf{X}},\,
    \{\mathbf{a}_j\}_{j\in\mathcal{J}}:\\
    |\mathcal{J}|\le M,\,
    \mathbb{E}\|\mathbf{X}-\mathbf{a}_J\|^2\le D
    }
    }
    I(\mathbf{X};J),
\end{equation}
with the convention \(R_M^{ }(D)=+\infty\) if the constraint set is
empty.
Every pair $(P_{J\mid\mathbf{X}},\mathbf{a}_j)$ in the constraint set of \eqref{eq:rdf_M} induces the test channel $P_{\widehat{\mathbf{X}} \mid \mathbf{X}=\mathbf{x}}=\sum_j P_{J\mid\mathbf{X}}(j\mid \mathbf{x})\delta_{\mathbf{a}_j}$ whose output takes at most \(M\) values. We call the pair, or the induced channel, an SSQ test channel of size at most $M$. Since $\widehat{\mathbf{X}}=\mathbf{a}_J$ is a function of $J$, $I(\mathbf{X};J)\geq I(\mathbf{X};\widehat{\mathbf{X}})$, with equality when $j\mapsto \mathbf{a}_j$ is injective.
If two codewords coincide, the corresponding symbols of \(\mathcal{J}\)
may be merged without changing the induced channel or the distortion.
Hence, \eqref{eq:rdf_M} is the infimum of $I(\mathbf{X};\widehat{\mathbf{X}})$ over test channels whose output support is at most $M$ points, which is equivalent to the output-alphabet-constrained RDF of \cite{FinamorePearlman1980}.

\begin{lemma}[Centroid reconstruction]
\label{lem:centroid}
Assume \(\mathbb{E}\|\mathbf{X}\|^2<\infty\), and let \(J\) be a
finite-valued random variable jointly distributed with \(\mathbf{X}\).
Among all codebooks \(\{\mathbf{a}_j\}\subset\mathbb{R}^d\),
the distortion \(\mathbb{E}\|\mathbf{X}-\mathbf{a}_J\|^2\) is minimized if
and only if
\(\mathbf{a}_j=\mathbb{E}[\mathbf{X}\mid J=j]\)
for every \(j\) with \(P(J=j)>0\).
In that case,
\begin{equation}
\label{eq:centroid_mse}
    \mathbb{E}\|\mathbf{X}-\mathbf{a}_J\|^2
    =
    \mathbb{E}\|\mathbf{X}\|^2
    -
    \mathbb{E}\|\mathbb{E}[\mathbf{X}\mid J]\|^2.
\end{equation}
In particular, the infimum in \eqref{eq:rdf_M} may be taken over
conditional kernels \(P_{J\mid\mathbf{X}}\) with
\(|\mathcal{J}|\le M\) and
\(\mathbb{E}\|\mathbf{X}-\mathbb{E}[\mathbf{X}\mid J]\|^2\le D\).
\end{lemma}
\begin{proof}
    See Appendix~\ref{app:centroid}.
\end{proof}

The cardinality constraint on \(J\) is the only restriction in
\eqref{eq:rdf_M} that is absent from the Shannon RDF. As \(M\) increases,
the feasible classes are nested, but it remains to determine whether
their union is sufficiently rich to recover the unrestricted RDF. Following \cite{FinamorePearlman1980}, we include the
following lemma for completeness, showing that the finite-alphabet penalty vanishes as
\(M\to\infty\).

\begin{lemma}[SSQ approximation of the RDF]
\label{lem:RM_limit}
Assume \(\mathbb{E}\|\mathbf{X}\|^2<\infty\). For every integer \(M\ge1\)
and every \(D>0\),
\begin{equation}
\label{eq:RM_mono}
    R^{ }(D)
    \le
    R_{M+1}^{ }(D)
    \le
    R_M^{ }(D),
\end{equation}
with the convention that both sides are \(+\infty\) if the corresponding
constraint set is empty.
Moreover, for every \(D\in(0,D_{\max})\),
\begin{equation}
\label{eq:RM_limit}
    \lim_{M\to\infty}R_M^{ }(D)
    =
    R^{ }(D).
\end{equation}
\end{lemma}
\begin{proof}
    See Appendix~\ref{app:centroid}.
\end{proof}

Lemma~\ref{lem:RM_limit} recovers the Shannon RDF by allowing the
reconstruction alphabet to grow. We now keep \(M\) finite, as required
for an SSQ construction with a fixed codebook. Assignments and codebooks
attaining or approximating \eqref{eq:rdf_M} can be computed by the
BA/deterministic-annealing iteration described in
Appendix~\ref{app:ba}. We now study the rate-distortion behavior of
fixed-cardinality SSQ, beginning at the low rate region near zero.

\subsection{Zero-Rate Boundary}
\label{sec:zero_rate}

Without loss of generality, assume \(\mathbb{E}[\mathbf{X}]=\mathbf{0}\),
otherwise replace \(\mathbf{X}\) by \(\mathbf{X}-\boldsymbol{\mu}\), which
leaves \(R^{ }\) and \(D_{\max}\) unchanged.
Define the directional sub-Gaussian proxy
\cite{BoucheronLugosiMassart2013}
\begin{equation}
    \tau^2
    :=
    \sup_{\mathbf{u}\ne\mathbf{0}}
    \frac{2}{\|\mathbf{u}\|^2}
    \log\mathbb{E}\bigl[e^{\mathbf{u}^{\mathsf{T}}\mathbf{X}}\bigr].
\end{equation}
If \(\mathbf{X}\) is a finite GMM with
\(\sum_{k=1}^{K}\pi_k\boldsymbol{\mu}_k=\mathbf{0}\), each component has a
Gaussian moment generating function, and therefore
\begin{equation}
    \tau^2
    =
    \sup_{\mathbf{u}\ne\mathbf{0}}
    \frac{2}{\|\mathbf{u}\|^2}
    \log
    \Biggl[
    \sum_{k=1}^{K}
    \pi_k
    \exp\Bigl(
    \mathbf{u}^{\mathsf{T}}\boldsymbol{\mu}_k
    +
    \tfrac12\mathbf{u}^{\mathsf{T}}\boldsymbol{\Sigma}_k\mathbf{u}
    \Bigr)
    \Biggr].
\end{equation}

We now derive the first-order asymptotic expansion for \(R^{ }\) and \(R_M^{ }\) near zero rate.

\begin{theorem}[Multivariate zero-rate boundary]
\label{thm:zero_rate_vec}
Let \(\delta:=D_{\max}-D\).
Then
\begin{equation}
    R^{ }(D_{\max}-\delta)
    =
    \frac{\delta}{2\tau^2}+o(\delta),
    \qquad\delta\downarrow0,
\end{equation}
and the same expansion holds for every \(R_M^{ }(D_{\max}-\delta)\) with
\(M\ge2\).
Moreover,
\begin{equation}
    R^{ }(D)
    \ge
    \frac{D_{\max}-D}{2\tau^2}
\end{equation}
for every \(D\in[0,D_{\max}]\).
\end{theorem}

\begin{proof}
The proof has three steps. First, a relative-entropy mean inequality gives a global
lower bound on both the Shannon RDF and its finite-alphabet restriction.
Second, it is shown that a binary test channel attains this bound to first order near
\(D=D_{\max}\). Finally, the monotonicity in \eqref{eq:RM_mono} extends the
binary result to every \(M\ge2\). The variational and truncation arguments
used below are established in Appendix~\ref{app:zero}.

\textit{Lower bound.}
Let \(P_{\widehat{\mathbf{X}}\mid\mathbf{X}}\) be any test channel with
distortion at most \(D\). The claim is immediate if
\(I(\mathbf{X};\widehat{\mathbf{X}})=\infty\); otherwise, let
\(Q_{\widehat{\mathbf{x}}}
:=P_{\mathbf{X}\mid\widehat{\mathbf{X}}=\widehat{\mathbf{x}}}\).
The conditional-mean property gives
\[
    D_{\max}-D
    \le
    \mathbb{E}\|
        \mathbb{E}[\mathbf{X}\mid\widehat{\mathbf{X}}]\|^2.
\]
For \(P_{\widehat{\mathbf{X}}}\)-almost every
\(\widehat{\mathbf{x}}\), applying the relative-entropy mean inequality
\eqref{eq:tau_kl} to \(Q_{\widehat{\mathbf{x}}}\) and then integrating
yields
\begin{align}
    D_{\max}-D
    &\le
    2\tau^2
    \int
    D_{\mathrm{KL}}(Q_{\widehat{\mathbf{x}}}\Vert P_{\mathbf{X}})
    P_{\widehat{\mathbf{X}}}(d\widehat{\mathbf{x}})
    \nonumber\\
    &=
    2\tau^2I(\mathbf{X};\widehat{\mathbf{X}}).
\end{align}
Taking the infimum over all admissible test channels proves
\begin{equation}
\label{eq:zero_global_lower}
    R^{ }(D)\ge\frac{D_{\max}-D}{2\tau^2},
    \qquad 0\le D\le D_{\max}.
\end{equation}

\textit{Binary upper bound.}
Fix \(\varepsilon\in(0,\tau^2)\). The variational characterization
\eqref{eq:tau_var} and the bounded-likelihood-ratio approximation in
Appendix~\ref{app:zero} provide a measure
\(Q\ll P_{\mathbf{X}}\) such that, with
\[
    \mathbf{m}:=\mathbb{E}_Q[\mathbf{X}],
    \qquad
    K_Q:=D_{\mathrm{KL}}(Q\Vert P_{\mathbf{X}}),
\]
\begin{equation}
\label{eq:bounded_Q_ratio}
    \frac{\|\mathbf{m}\|^2}{2K_Q}
    \ge\tau^2-\varepsilon,
    \qquad
    L:=\frac{dQ}{dP_{\mathbf{X}}}\le B
\end{equation}
for some finite \(B>1\). In particular, \(K_Q>0\) and
\(\mathbf{m}\ne\mathbf{0}\).

For \(0<q\le B^{-1}\), define
\[
    P(J=1\mid\mathbf{X})=qL,
    \qquad
    P(J=0\mid\mathbf{X})=1-qL.
\]
Then \(P(J=1)=q\), \(P_{\mathbf{X}\mid J=1}=Q\), and
\[
    \mathbb{E}[\mathbf{X}\mid J=1]=\mathbf{m},
    \qquad
    \mathbb{E}[\mathbf{X}\mid J=0]
    =-\frac{q}{1-q}\mathbf{m}.
\]
Using these conditional means as the two reconstruction codewords gives
\begin{equation}
\label{eq:binary_deficit}
    D_{\max}-\mathbb{E}\|
        \mathbf{X}-\mathbb{E}[\mathbf{X}\mid J]\|^2
    =
    \frac{q}{1-q}\|\mathbf{m}\|^2.
\end{equation}
For any sufficiently small \(\delta>0\), choose
\[
    q_\delta:=\frac{\delta}{\|\mathbf{m}\|^2+\delta}.
\]
Then \(q_\delta\le B^{-1}\), and \eqref{eq:binary_deficit} shows that the
distortion is exactly \(D_{\max}-\delta\). The rare-binary expansion established
in Appendix~\ref{app:zero} gives
\begin{align}
    I(\mathbf{X};J)
    &=
    q_\delta K_Q+O(q_\delta^2)
    \nonumber\\
    &=
    \frac{K_Q}{\|\mathbf{m}\|^2}\delta+O(\delta^2)
    \le
    \frac{\delta}{2(\tau^2-\varepsilon)}+O(\delta^2),
\end{align}
where the last inequality follows from
\eqref{eq:bounded_Q_ratio}. Consequently,
\[
    \limsup_{\delta\downarrow0}
    \frac{R_2^{ }(D_{\max}-\delta)}{\delta}
    \le
    \frac{1}{2(\tau^2-\varepsilon)}.
\]
Letting \(\varepsilon\downarrow0\) and combining this bound with the
lower bound \eqref{eq:zero_global_lower} yields
\[
    R_2^{ }(D_{\max}-\delta)
    =
    \frac{\delta}{2\tau^2}+o(\delta).
\]
Finally, for every \(M\ge2\), \eqref{eq:RM_mono} gives
\[
    R^{ }(D_{\max}-\delta)
    \le
    R_M^{ }(D_{\max}-\delta)
    \le
    R_2^{ }(D_{\max}-\delta).
\]
Applying \eqref{eq:zero_global_lower} to the left-hand side completes the
sandwich and proves the asserted expansion for \(R^{ }\) and every
\(R_M^{ }\), \(M\ge2\).
\end{proof}

In particular, \(M=2\) already matches the coefficient of \(\delta\) in
Theorem~\ref{thm:zero_rate_vec} and enlarging the SSQ codebook does not change
the first-order expansion as \(D\uparrow D_{\max}\).
While \(\tau^2\) is typically difficult to compute, as
\(\mathbf{u}\to\mathbf{0},
\log\mathbb{E}\bigl[e^{\mathbf{u}^{\mathsf{T}}\mathbf{X}}\bigr]
=
\tfrac12\mathbf{u}^{\mathsf{T}}\boldsymbol{\Sigma}_{\mathbf{X}}\mathbf{u}
+o(\|\mathbf{u}\|^2).
\)
Evaluating along a leading eigenvector and letting
\(\|\mathbf{u}\|\to0\) yields
\(\tau^2\ge\lambda_{\max}(\boldsymbol{\Sigma}_{\mathbf{X}})\).
A useful upper bound on the slope is therefore
\(1/(2\lambda_{\max}(\boldsymbol{\Sigma}_{\mathbf{X}}))\).
In the scalar case this is \(1/(2D_{\max})\).

\subsection{Limitations of SSQ}
\label{sec:floor}
As the rate increases, suboptimality of SSQ emerges because the
reproduction codebook is finite. We record and discuss two overlapping regions in this subsection: a strict inequality between the
SSQ and Shannon RDF wherever the Shannon
reproduction is non-atomic, and a geometric floor below which no
SSQ codebook of size \(M\) is feasible.

We first note a corollary in the high-rate regime, reusing the proof of Theorem \ref{thm:grm}.

\begin{corollary}[Strict inequality of the SSQ rate]
\label{cor:gaps}
Let \(M<\infty\) and \(D\in(0,D_{\max})\) satisfy \(R_M^{ }(D)<\infty\).
Then \(R_M^{ }(D)>R^{ }(D)\) in each of the following cases:
\begin{enumerate}
    \item \(\mathbf{X}\) is Gaussian;
    \item \(R_k(D)=\mathrm{SLB}_k(D)\) for every \(k\), and the
    equality-case reproduction of \eqref{eq:slb} is non-atomic.
\end{enumerate}

\end{corollary}

\begin{proof}
The passage from a non-atomic Shannon reproduction to a strict inequality
of infima is given in Appendix~\ref{app:gaps}.
It remains to identify that reproduction.

1)
A Gaussian RDF-achieving reproduction is Gaussian on the active
eigenspace \cite{Berger1971,CoverThomas}, as we established in
Appendix~\ref{app:rwf_vec}. That subspace is nontrivial for \(D<D_{\max}\), so
the reproduction is non-atomic.

2)
Theorem~\ref{prop:comp_slb} gives \(R^{ }(D)=\mathrm{SLB}(D)\). Every
RDF-achieving pair therefore uses the equality-case reproduction of
\eqref{eq:slb}, which is non-atomic by assumption.

\end{proof}

Thus, in the high-rate region where GRM is optimal due to Theorem \ref{thm:grm}, SSQ is strictly suboptimal for finite $M$, and the optimality is only achieved at the infinite limit by Lemma \ref{lem:RM_limit}.

While Corollary \ref{cor:gaps} applies to arbitrary $M$, there also exists a distortion floor for each finite $M$, as a finite SSQ codebook prohibits arbitrarily small distortion, regardless of rate. We present this result, which is largely classical, for completeness. The SSQ distortion-rate function can be written as
\begin{equation}
    D_M^{ }(R)
    :=
    \inf
    \bigl\{
    \mathbb{E}\|\mathbf{X}-\mathbf{a}_J\|^2
    :\
    I(\mathbf{X};J)\le R,\ |\mathcal{J}|\le M
    \bigr\}
\end{equation}
for the dual of \(R_M^{ }(D)\).

\begin{lemma}[Distortion floor]
\label{lem:floor}
Let \(\mathbf{X}\) be a nondegenerate finite GMM on \(\mathbb{R}^d\) and
let \(M<\infty\). Set
\begin{equation}
    \label{eq:floor_expr}
    D_M^{\star}
    :=
    \inf_{\mathbf{a}_1,\ldots,\mathbf{a}_M\in\mathbb{R}^d}
    \mathbb{E}
    \Bigl[
    \min_{1\le j\le M}
    \|\mathbf{X}-\mathbf{a}_j\|^2
    \Bigr].
\end{equation}
Then
\begin{enumerate}
    \item \(D_M^{\star}>0\);
    \item \(R_M^{ }(D)=+\infty\) for \(D<D_M^{\star }\);
    \item \(R_M^{ }(D_M^{\star })\le\log M\);
    \item as \(M\to\infty\),
    \begin{equation}
        \label{eq:zador}
        D_M^{\star }
        \sim
        C_d(\mathbf{X})\,M^{-2/d},
    \end{equation}
    where
    \(C_d(\mathbf{X})
    :=
    G_d
    \int_{\mathbb{R}^d}
    f_{\mathbf{X}}(\mathbf{x})^{d/(d+2)}\,d\mathbf{x}\)
    and \(G_d\in(0,\infty)\) is Zador's constant
    \cite{Zador1982} (\(G_1=1/12\) \cite{PanterDite1951}).
\end{enumerate}
\end{lemma}
\begin{proof}
    See Appendix~\ref{app:floor}.
\end{proof}

One should note that there is no direct relationship between the distortion floor and conditions of Corollary \ref{cor:gaps}. On the other hand, while the floor vanishes as \(M\to\infty\), for applications of fitted GMM, $K$ may not be inherent and can be a controllable parameter as well. We therefore further study below the joint asymptotic relation, considering that $M,K$ jointly go to infinity. 

\begin{proposition}[Joint \(M,K\) vanishing of the floor]
\label{prop:MK}
Let \(\{\mathbf{X}^{(K)}\}_{K\ge1}\) be a sequence of \(K\)-component GMMs on
\(\mathbb{R}^d\) whose component eigenvalues lie in a fixed compact
interval \([\lambda_{\min},\lambda_{\max}]\subset(0,\infty)\), and let
\(M=M(K)\).
\begin{enumerate}
    \item If \(M/K\to\infty\), then
    \(D_M^{\star }(\mathbf{X}^{(K)})\to0\).
    \item If \(\pi_k=1/K\) and
    \(\|\boldsymbol{\mu}_k-\boldsymbol{\mu}_{k'}\|\ge 4\sqrt{\lambda_{\min}}\)
    for all \(k\neq k'\), then
    \(D_M^{\star }(\mathbf{X}^{(K)})
    =\Omega\bigl((K/M)^{2/d}\bigr)\).
    \item Under the hypotheses of (2), if \(M/K\) is bounded then
    \(\liminf_{K\to\infty}
    D_M^{\star }(\mathbf{X}^{(K)})>0\).
\end{enumerate}
\end{proposition}

\begin{proof}

1)
Set \(m:=\lfloor M/K\rfloor\). The hypothesis \(M/K\to\infty\) forces
\(m\to\infty\).
It is enough to exhibit one \(M\)-point codebook whose MSE tends to zero.
Allocate \(m\) codewords to each of the \(K\) components (using at most
\(M\) points) and take the union of the \(K\) codebooks.

Fix a Gaussian \(\mathcal{N}(\boldsymbol{\mu},\boldsymbol{\Sigma})\) with eigenvalues at
most \(\lambda_{\max}\).
Rotate to principal axes and form a product uniform scalar quantizer with
\(\ell:=\max\{\lfloor m^{1/d}\rfloor,1\}\) levels per coordinate on
\([-A,A]\), where \(A:=\sqrt{4\lambda_{\max}\log\ell}\), sending overflow
to the nearest endpoint.
Let \(\varphi\) and \(\Phi\) denote the standard normal density and cdf.
The Mills ratio \cite{Mills1926} \(R(x)=(1-\Phi(x))/\varphi(x)\) satisfies \(R(x)\le 1/x\)
for \(x>0\). Thus if
\(Z\sim\mathcal{N}(0,\sigma^2)\) and \(\alpha=A/\sigma\), then
\(\mathbb{E}[(Z-A)_+^2]\le\sigma^2\varphi(\alpha)/\alpha\).
The choice of \(A\) gives \(\alpha\ge 2\sqrt{\log\ell}\), hence the
two-sided overflow second moment is
\(O\bigl(\lambda_{\max}(\log\ell)\,\ell^{-2}\bigr)\) uniformly for
\(\sigma^2\in(0,\lambda_{\max}]\). The interior granular error is at most
\((A/\ell)^2=O\bigl(\lambda_{\max}(\log\ell)\,\ell^{-2}\bigr)\).
Summing \(d\) coordinates, the vector MSE is
\(O\bigl(\lambda_{\max}(1+\log m)\,m^{-2/d}\bigr)\), with a constant
depending only on \(d\).
The union codebook therefore satisfies
\[
D_M^{\star }(\mathbf{X}^{(K)})
=
O\bigl(\lambda_{\max}(1+\log(M/K))(M/K)^{-2/d}\bigr)
\to0.
\]

2)
Set \(r:=\sqrt{\lambda_{\min}}\).
On the ball \(B_k\) of radius \(r\) about \(\boldsymbol{\mu}_k\), the component
density is bounded below by a positive constant depending only on
\(\lambda_{\max}\) and \(d\).
With \(\pi_k=1/K\), the mixture density is therefore bounded below by
\(c/K\) on each \(B_k\), for some \(c>0\) independent of \(K\).
The separation \(\|\boldsymbol{\mu}_k-\boldsymbol{\mu}_{k'}\|\ge 4r\) makes the
enlarged balls of radius \(2r\) pairwise disjoint, so each codeword lies
in at most one such neighborhood.

Let \(m_k\) be the number of codewords in the \(2r\)-neighborhood of
\(\boldsymbol{\mu}_k\), so \(\sum_k m_k\le M\).
If some \(m_k=0\), that ball contributes a \(K\)-independent positive
amount and the lower bound only increases.
Repeating the volume argument of Lemma~\ref{lem:floor} on each \(B_k\)
with \(m_k\ge1\) therefore yields
\[
D_M^{\star }(\mathbf{X}^{(K)})
\ge
\frac{c}{K}
\sum_{k=1}^{K}
\kappa_d v^{1+2/d}\, m_k^{-2/d},
\]
where \(v\) is the volume of a ball of radius \(r\).
The right-hand side, subject to \(\sum_k m_k\le M\), is minimized when
\(m_k=M/K\), and that value scales with \(\Omega\bigl((K/M)^{2/d}\bigr)\).

3)
If \(M/K\) is bounded, then \(K/M\) is bounded away from zero, so the
right-hand side of (ii) is bounded away from zero.
\end{proof}

Therefore, we retrieve the vanishing floor whenever $M>>K$, which is sufficient for general GMMs and necessary for equal-weight well-separated GMMs. Combined with
Lemma~\ref{lem:RM_limit}, the convergence to rate-distortion function is guaranteed for each fixed source, while in a low-distortion region rate-distortion function remains non-attainable, as suggested by Corollary~\ref{thm:grm}.

\begin{figure*}[t]
    \centering
    \includegraphics[width=\textwidth]{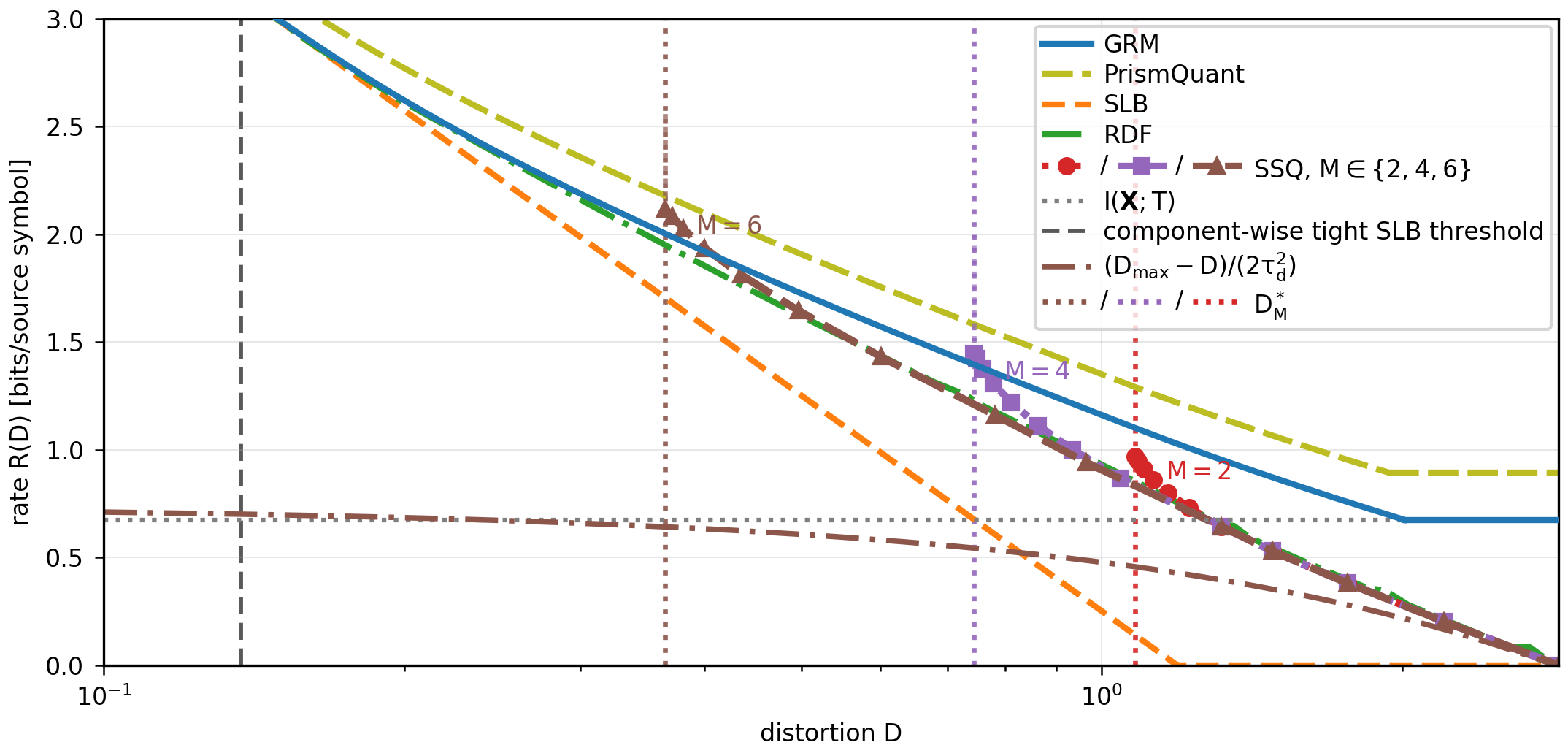}
    \caption{Two-component GMM on \(\mathbb{R}^2\) with
    \(\pi_1=0.3\), \(\pi_2=0.7\),
    \(\boldsymbol{\mu}_1=(-1,0)^{\mathsf{T}}\),
    \(\boldsymbol{\mu}_2=(1,0)^{\mathsf{T}}\),
    \(\boldsymbol{\Sigma}_1=\mathrm{diag}(0.1,1)\),
    and
    \(\boldsymbol{\Sigma}_2=\bigl(\begin{smallmatrix}2&0.8\\0.8&0.4\end{smallmatrix}\bigr)\).
    Rates are in bits per source symbol.
    Shown are GRM \eqref{eq:R_GRM_vec},
    PrismQuant \cite{PrismQuant2026} (MAP label entropy-coded at
    \(H(\widehat{C})\)),
    the SLB \eqref{eq:slb},
    a BA estimate of \(R(D)\),
    and SSQ information rates \(R_M(D)\) for \(M\in\{2,4,6\}\).
    Reference lines indicate the GRM label cost \(I(\mathbf{X};T)\),
    the fully active threshold of Corollary~\ref{thm:grm}, the zero-rate
    bound of Theorem~\ref{thm:zero_rate_vec}, and the high-rate SSQ
    plateaus on the discretized source.}
    \label{fig:bimodal_2d}
\end{figure*}

\section{Numerical Results}
\label{sec:numerics}

In this section, we present three complementary numerical studies. The first compares
GRM and SSQ with the Shannon RDF and related bounds across the high/low-rate regimes of a fixed GMM. The second examines the
finite-blocklength overhead of channel simulation and compares stochastic and
deterministic quantization. The third varies the numbers of mixture
components and reconstruction points to test the distortion-floor
scaling of Proposition~\ref{prop:MK}. Rates, where plotted, are reported
in bits per source symbol, rather than the nats used in the analytical
sections.

\subsection{High- and Low-Rate Behavior}

Fig.~\ref{fig:bimodal_2d} illustrates the high/low-rate
constructions of Sections~\ref{sec:grm}--\ref{sec:lowrate} for the
two-component GMM specified in its caption.

We approximate \(h(\mathbf{X})\) and \(I(\mathbf{X};T)\) by trapezoidal
quadrature of the mixture density and of the posterior
\(P_{T\mid\mathbf{X}}\) on a uniform \(41\times 41\) product grid.
Each axis is truncated to a bounded interval covering six component-wise
standard deviations about every mean, so the tails of the unbounded
mixture are neglected (cf.\ \cite{AkyolRose2014}). From these quantities
we evaluate the SLB \eqref{eq:slb} and the closed-form GRM expression
\eqref{eq:R_GRM_vec}. This comparison is for the asymptotic infinite blocklength regime:
the rate of sampling from the posterior kernel in GRM is \(I(\mathbf{X};T)\), without a
finite-blocklength overhead. The PrismQuant baseline
\cite{PrismQuant2026} uses the hard MAP label
\(\widehat{C}=\arg\max_k P(T=k\mid\mathbf{X})\), entropy-coded at
\(H(\widehat{C})\), followed by the component Karhunen-Lo\`eve transform
of \(\boldsymbol{\Sigma}_{\widehat{C}}\) and independent scalar codes. The
residual \(\mathbf{X}\mid\widehat{C}=k\) is the mixture restricted to a
MAP cell. We therefore
evaluate the KLT coefficients of \(\mathbf{X}-\boldsymbol{\mu}_k\) on that
cell by quadrature and charge each coordinate the scalar Gaussian RDF
of its empirical second moment, with a single water level across
coordinates and cells.

A numerical estimate of the Shannon RDF \(R(D)\) is computed by the
BA algorithm \cite{Blahut1972} after discretizing the source
on a coarser \(25\times 25\) grid of the same support, retaining the
dominant probability mass.
SSQ curves for \(M\in\{2,4,6\}\) are the rates \(R_M(D)\) obtained from \eqref{eq:rdf_M}. A reconstruction index
\(J\) of cardinality \(M\) is optimized by the deterministic-annealing /
BA iteration of Appendix~\ref{app:ba} on that same discrete
source, and is likewise charged its simulation rate \(I(\mathbf{X};J)\).
Thus the Shannon RDF and SSQ curves use a common discretized source, and as such both serve as numerical approximations to their continuous-source
counterparts.

The directional proxy \(\tau^2\) in Theorem~\ref{thm:zero_rate_vec} is
evaluated from the closed-form log-moment generating function of the GMM
by a search over directions \(\mathbf{u}\in\mathbb{R}^2\).
The vertical line labeled as the component-wise tight SLB threshold is
the fully active boundary \(D=d\min_{k,m}\lambda_{k,m}\) of
Corollary~\ref{thm:grm}.
The dotted vertical lines are the high-rate nearest-neighbor plateaus
of the corresponding SSQ curves on that discretized source
(Lemma~\ref{lem:floor}).

\begin{figure*}[t]
    \centering
    \includegraphics[width=\textwidth]{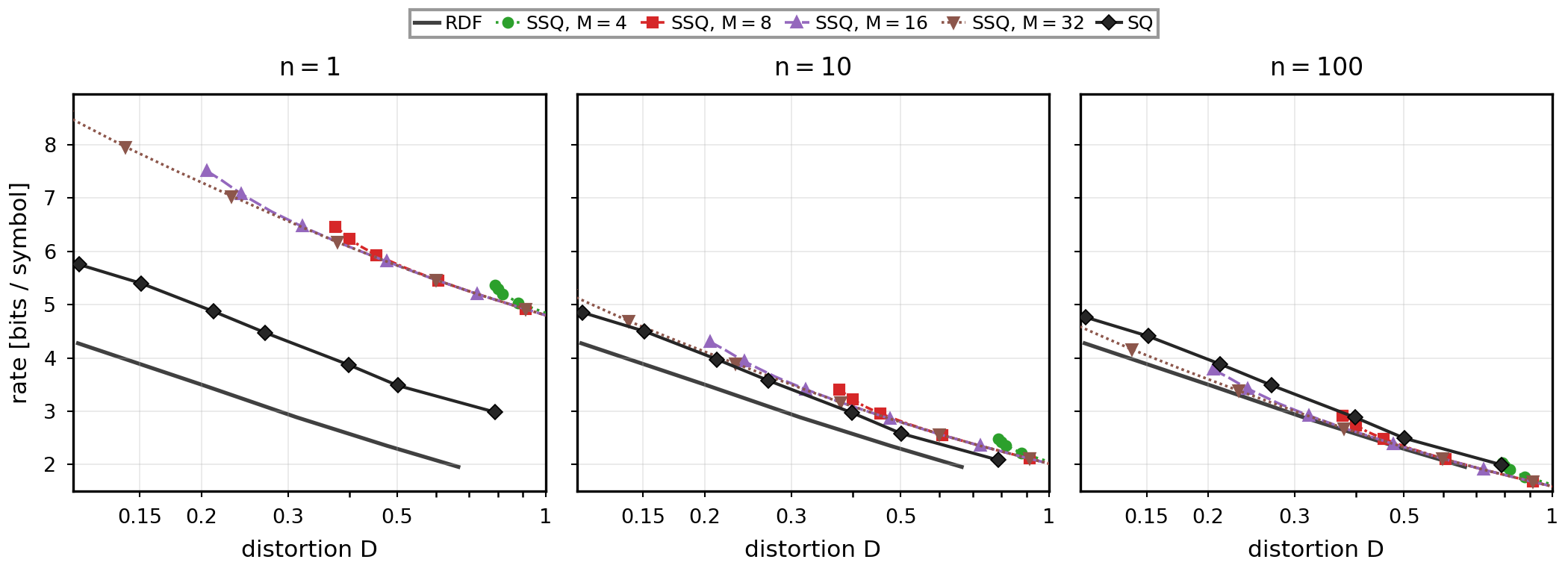}
    \caption{Variable-length finite-blocklength comparison on a
    three-component GMM on \(\mathbb{R}^2\) with
    \(\pi=(0.42,0.33,0.25)\),
    \(\boldsymbol{\mu}_1=(-1.7,0.6)^{\mathsf{T}}\),
    \(\boldsymbol{\mu}_2=(1.8,1.1)^{\mathsf{T}}\),
    \(\boldsymbol{\mu}_3=(0.15,-1.8)^{\mathsf{T}}\),
    \(\boldsymbol{\Sigma}_1=\bigl(\begin{smallmatrix}0.16&0.05\\0.05&0.85\end{smallmatrix}\bigr)\),
    \(\boldsymbol{\Sigma}_2=\bigl(\begin{smallmatrix}1.20&-0.40\\-0.40&0.30\end{smallmatrix}\bigr)\),
    and
    \(\boldsymbol{\Sigma}_3=0.38\,I_2\).
    Panels correspond to blocklengths \(n\in\{1,10,100\}\).
    The solid curve is a BA estimate of the Shannon RDF
    \(R(D)\). The remaining curves are SSQ rates
    \(R_M^{}\) for \(M\in\{4,8,16,32\}\), each lifted by the one-shot channel simulation bound \eqref{eq:rec_fin}--\eqref{eq:rec_fin_nlarge}.
    Diamonds (from right to left) are SQ operating points at codebook sizes
    \(\{4,6,8,12,16,24,32\}\), charged \(H(J)+1/n\) by
    \eqref{eq:sq_fin}.
    Rates are in bits per source symbol.}
    \label{fig:trimodal}
\end{figure*}

To the left of the fully active threshold, GRM coincides with the SLB, in
agreement with Corollary~\ref{thm:grm}, while the BA estimate
closely follows these curves up to discretization error.
PrismQuant lies above GRM, where the gap reflects both the
larger label cost \(H(\widehat{C})\) relative to \(I(\mathbf{X};T)\) and
the mismatch between the MAP-cell residual and the component Gaussian.
Beyond the threshold, GRM remains an achievable upper bound and
numerically separates from both the SLB and RDF estimates.
Near \(D_{\max}\) the SSQ schemes track the RDF, consistent
with the first-order expansion of Theorem~\ref{thm:zero_rate_vec}
(overlaid as \((D_{\max}-D)/(2\tau^2)\)). As the distortion decreases, the SSQ curves depart from the RDF and approach their finite-codebook distortion floors, as predicted by Lemma~\ref{lem:floor}.
The dotted horizontal line is the GRM label cost \(I(\mathbf{X};T)\), below which
GRM cannot operate.

\subsection{Finite-Blocklength Regime}

The preceding comparison charges each stochastic kernel at its
single-letter mutual information bound which is achievable only in the infinite-blocklength regime. We next account for the overhead of
variable-length prefix coding at finite blocklength and compare SSQ with
deterministic SQ.

Fig.~\ref{fig:trimodal} compares SSQ and SQ with a
numerical RDF for a three-component anisotropic GMM in \(\mathbb{R}^2\),
with parameters listed in the caption.
A BA estimate of \(R(D)\) is computed on a pruned
\(25\times 25\) discretization of the source. The SSQ curves
\(R_M^{}\) of \eqref{eq:rdf_M} for \(M\in\{4,8,16,32\}\) and the SQ
operating points use a pruned \(29\times 29\) discretization of the
same support. Before finite-blocklength overhead is
added, the stochastic kernels are charged at \(I(\mathbf{X};J)\), whereas
the deterministic SQ index is charged at rate \(H(J)\). The denser set of SQ
codebook sizes listed in the caption provides additional operating-point
coverage for the deterministic tradeoff.

We evaluate variable-length prefix codes at
\(n\in\{1,10,100\}\), and non-prefix one-to-one encodings are not included.
Distortion is the expected per-letter MSE of the memoryless
reconstruction for both schemes. For each SSQ curves, the
discrete index generated by the \(n\)-fold product kernel is channel simulated exactly. For each
realization of the common randomness, the encoder uses a binary prefix
code \cite{FlamichGunduzREC,LiLayered2025}. We denote the one-shot prefix-free channel simulation
bound as \(\overline{L}(\,\cdot\,)\). Applying to the block mutual information \(nI\)
gives
\begin{equation}
\label{eq:rec_fin}
    \overline{R}_n
    =
    \frac{\overline{L}(nI)}{n}
    =
    I
    +
    \frac{\overline{L}(nI)-nI}{n}.
\end{equation}
For \(\overline{L}\), we use the one-shot bounds of
\cite{LiLayered2025}. Over the range plotted in Fig.~\ref{fig:trimodal},
the \(n=1\) curves use
\begin{equation}
\label{eq:rec_fin_n1}
    \overline{L}(s)
    =
    s+\log_2(s+3.4)+1,
\end{equation}
whereas the \(n\in\{10,100\}\) curves use
\begin{equation}
\label{eq:rec_fin_nlarge}
    \overline{L}(s)
    =
    s+\log_2(s+5.22)+0.76,
\end{equation}
where \(s=nI\). These are the respective tighter bounds at the block
mutual-information values encountered in the experiment.
For SQ, \(J^n\) instead is an i.i.d.\ discrete source of entropy \(H(J)\).
By the classical prefix-code bound
\cite[Thm.~5.4.1]{CoverThomas},
\begin{equation}
\label{eq:sq_fin}
    H(J)
    \le
    \frac{\mathbb{E}[L_n]}{n}
    <
    H(J)+\frac{1}{n},
\end{equation}
where the \(+1\) bit bounds the overhead of the whole block.
Fig.~\ref{fig:trimodal} plots the upper end \(H(J)+1/n\).

At \(n=1\), the one-shot simulation overhead,
\(\log_2(I+3.4)+1\), places the SSQ curves several bits above
their information rates, whereas SQ pays at most the one-bit prefix
overhead in \eqref{eq:sq_fin}. Over the plotted range, the normalized one-shot simulation
overhead is below one bit per letter at \(n=10\) and approximately one
tenth of a bit per letter at \(n=100\). Accordingly, at \(n=100\),
SSQ curves lie close to \(R(D)\), while SQ remains
separated by the deterministic nearest-neighbor gap
\(H(J)-R_M(D)\).

We note that whereas SSQ becomes a better choice as $n$ increases, since Rényi $\infty$-divergence of the $n$-fold product is linear in $n$, the computational complexity of SSQ is constant in distortion but increases exponentially with $n$ due to channel simulation. On the other hand, the complexity of SQ is linear in $n$ and polynomial in $1/D$. Hence, where energy and latency are sensitive, SQ may remain viable in the high-distortion regime.

\subsection{Distortion-Floor Scaling}

The first experiment displayed finite-codebook floors for one fixed GMM.
We now vary both the number \(K\) of mixture components and the codebook
size \(M\) to examine their joint scaling in
Proposition~\ref{prop:MK}.

\begin{figure}[t]
    \centering
    \subfloat[]{\includegraphics[width=\columnwidth]{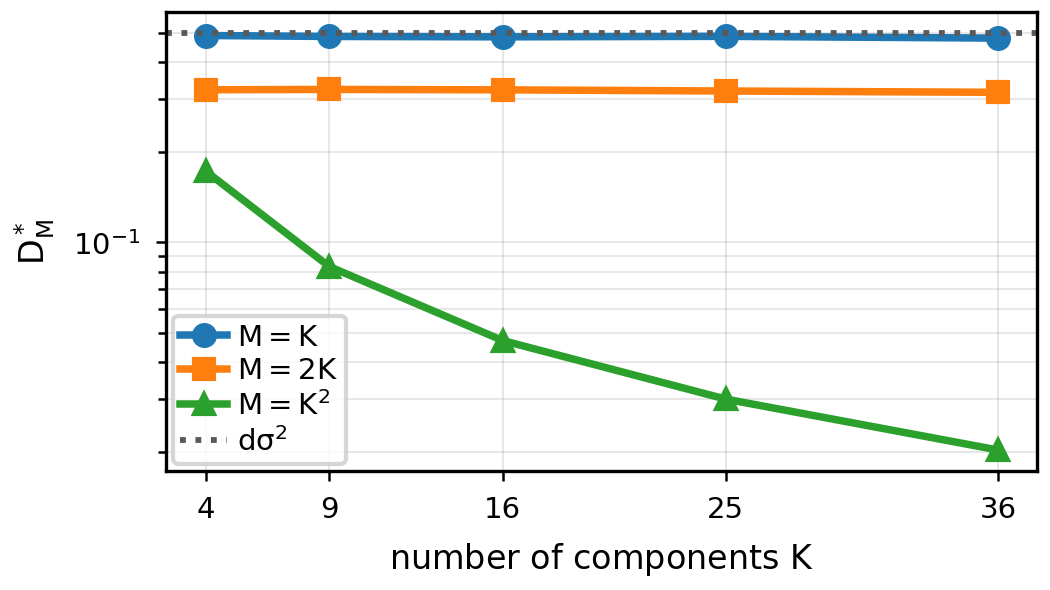}%
    \label{fig:prop_mk_a}}
    \\
    \subfloat[]{\includegraphics[width=\columnwidth]{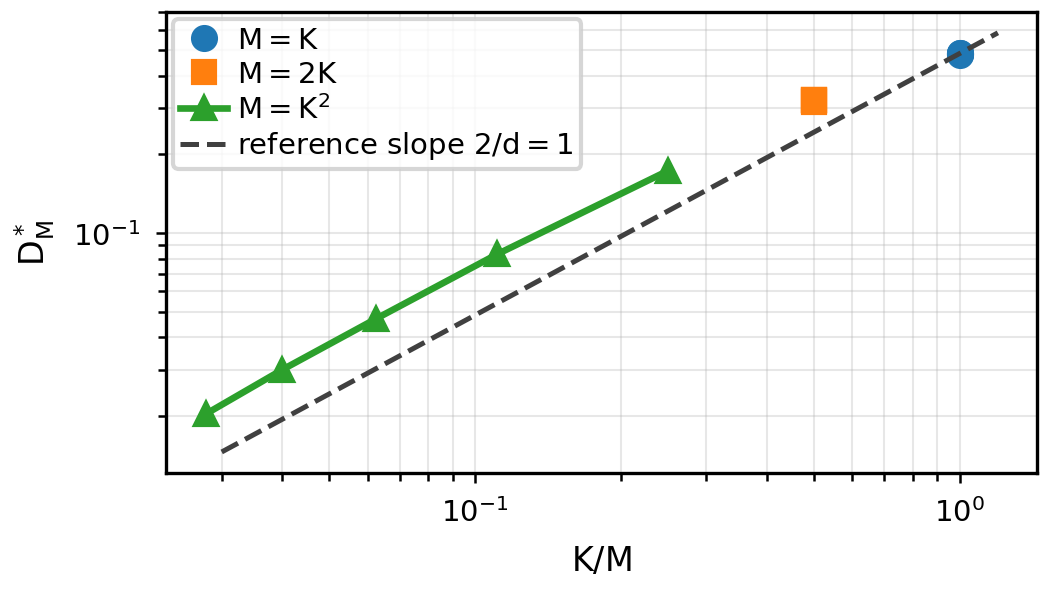}%
    \label{fig:prop_mk_b}}
    \caption{Equal-weight isotropic GMMs on \(\mathbb{R}^2\) with
    component variance \(\sigma^2=0.25\) and means on a square lattice of
    spacing \(2.5\ge 4\sqrt{\lambda_{\min}}\).
    Estimated floors \(D_M^{\star}\) for \(K\in\{4,9,16,25,36\}\) and
    the schedules \(M=K\), \(M=2K\), and \(M=K^2\).
    (a)~Floor versus \(K\): when \(M/K\) remains bounded, the floor
    stays away from zero; \(M/K\to\infty\) drives it to zero.
    (b)~Floor versus \(K/M\), with a reference of slope \(2/d=1\),
    consistent with the \(\Omega\bigl((K/M)^{2/d}\bigr)\) scaling of
    Proposition~\ref{prop:MK}. In panel (b), the \(M=K\)
    (resp.\ \(M=2K\)) markers coincide across \(K\).}
    \label{fig:prop_mk}
\end{figure}

Fig.~\ref{fig:prop_mk} presents the Monte Carlo estimates of the
nearest-neighbor floor $D^*_M$ in \eqref{eq:floor_expr} for the GMM family in its
caption. Each estimate uses at least
\(\max\{40000,40M\}\) source samples and iterative nearest-neighbor
codebook optimization. The component eigenvalues equal
\(\lambda_{\min}=0.25\), and the lattice spacing \(2.5\) exceeds
\(4\sqrt{\lambda_{\min}}=2\); hence the family satisfies the
equal-weight separation hypothesis of Proposition~\ref{prop:MK}.

Panel~\ref{fig:prop_mk_a} shows the dependence of $D^*_M$ on \(K\). For \(M=K\),
the estimated floor remains within \(4.1\%\) of
\(d\sigma^2=0.5\), which is consistent since reconstruction at \(\boldsymbol{\mu}_k\) gives
\(\mathbb{E}\bigl[\|\mathbf{X}-\boldsymbol{\mu}_k\|^2\bigm|T=k\bigr]
=\operatorname{tr}(\sigma^2 I_d)=d\sigma^2\), while the component separation
makes cross-component assignments rare. For \(M=2K\), the floor remains
near \(0.32\). Both schedules have bounded \(M/K\) and remain bounded
away from zero, as asserted by claim 3) of Proposition \ref{prop:MK}. In contrast, \(M=K^2\) gives
\(M/K=K\to\infty\), and the floor decreases from \(0.17\) at \(K=4\) to
\(0.020\) at \(K=36\), consistent with 1) in Proposition \ref{prop:MK}.

Panel~\ref{fig:prop_mk_b} displays the estimates against \(K/M\).
For \(d=2\), claim~2) of Proposition~\ref{prop:MK} gives the lower-bound
scaling \(\Omega(K/M)\). On the log-log axes, the \(M=K^2\) points are
nearly collinear with the reference slope \(1=2/d\). Since
\(K/M=1/K\) under this schedule, the observed behavior is approximately
\(D_{K^2}^{\star}\propto K^{-1}\), consistent with the exponent in
claim~2). Together, the two panels also illustrate the dichotomy between
claims~1) and~3), that the floor decreases when \(M/K\) diverges but remains
bounded away from zero when \(M/K\) stays bounded.

\section{Conclusion and Future Work}
\label{sec:discussion}

This paper studied MSE rate-distortion trade-off of finite mixtures, with an emphasis on GMMs, through two constructions adapted for the two opposite
distortion regimes. GRM exploits the latent mixture structure at high rate, whereas SSQ uses a finite reconstruction alphabet to approach the zero-rate boundary. Together, they clarify both the value and the cost
of retaining the mixture structure in an operational coding scheme.

GRM first conveys a sample from the posterior mixture label through channel simulation at rate $I(\mathbf{X}; T)$, and then applies a component-conditional Gaussian code. For a finite mixture, GRM meets the
SLB, and hence the RDF, if and only if every component meets its corresponding SLB. For a GMM, this criterion yields exact optimality throughout the fully active interval, i.e., whenever $D/d$ is below the smallest eigenvalue among all component covariances, so that reverse water-filling is active in every direction of all Gaussian components. Outside that interval, GRM remains an
achievable scheme, but not necessarily meeting the RDF. The posterior sampling cost of GRM sets a lower bound on its rate, which becomes provably suboptimal in the high-distortion regime.

We then propose SSQ for that regime. SSQ optimizes a finite reconstruction index directly; and
therefore, avoids the fixed posterior-label cost of GRM. Every codebook size \(M\ge2\) attains the first-order zero-rate expansion of the RDF, and \(R_M^{ }(D)\)
converges to the RDF as \(M\to\infty\). At any fixed finite \(M\), however, SSQ has a nearest-neighbor distortion floor and is strictly suboptimal whereever the Shannon reproduction is non-atomic, in particular where GRM meets SLB. Proposition~\ref{prop:MK}
further relates the vanishing of the distortion floor to the growth of \(M\)
relative to the number of mixture components.

Both constructions rely on exact channel simulation with unrestricted common randomness. Two natural extensions are the finite-blocklength implementation of GRM, where the Gaussian stage would be realized by lattice or dithered quantization, and the tradeoff that can be achieved by relaxing the exactness of the channel simulation stage.

Further theoretical questions include characterizing GRM outside the fully active interval and determining whether SSQ with some finite \(M\) can be exactly optimal on an interval adjacent to \(D_{\max}\). A more practical direction is to evaluate GRM in learned mixture-based codecs
by replacing hard latent-class assignment with simulated posterior sampling and measuring the resulting operational rate-distortion tradeoff.

\appendices

\section{Auxiliary Results for Gibbs Relabeling}
\label{app:grm}

\subsection{Scalar reverse water-filling}
\label{app:rwf_scalar}

We derive \eqref{eq:RWF_scalar}--\eqref{eq:Rcond_closed}.
Each map \(D_k\mapsto R_k(D_k)\) in \eqref{eq:Rk_scalar} is convex,
nonincreasing, and twice continuously differentiable on \((0,\sigma_k^2)\) with
\(R_k''(D_k)=1/(2D_k^2)>0\).
In particular \(R_k(D_k)=0\) for all \(D_k\ge\sigma_k^2\), so any
minimizer of \eqref{eq:Rcond_scalar} may be taken to satisfy
\(D_k\le\sigma_k^2\): increasing a coordinate beyond \(\sigma_k^2\) cannot
decrease the objective and consumes the linear budget.
The resulting program is strictly convex on the compact simplex
\(\{D_k\in(0,\sigma_k^2]:\sum_k\pi_k D_k=D\}\), hence has a unique minimizer,
and the KKT conditions are necessary and sufficient.
The Lagrangian is
\begin{equation}
    \mathcal{L}
    =
    \sum_{k=1}^{K}
    \pi_k\cdot\tfrac12\log(\sigma_k^2/D_k)
    +
    \lambda
    \Biggl(
    \sum_{k=1}^{K}\pi_k D_k-D
    \Biggr).
\end{equation}
Stationarity for an active coordinate (\(D_k<\sigma_k^2\)) yields
\(-\pi_k/(2D_k)+\lambda\pi_k=0\), hence \(D_k=1/(2\lambda)\).
Setting \(\theta:=1/(2\lambda)\) and enforcing \(D_k\le\sigma_k^2\) produces
\eqref{eq:RWF_scalar}.
The water level \(\theta>0\) is uniquely determined by
\eqref{eq:theta_scalar}: the left-hand side is continuous in \(\theta\),
equals \(0\) at \(\theta=0\), equals \(\sum_k\pi_k\sigma_k^2\) for
\(\theta\ge\max_k\sigma_k^2\), and is strictly increasing on
\([0,\max_k\sigma_k^2]\) because its derivative equals the total mass of
components with \(\sigma_k^2>\theta\), which is positive throughout that
interval.
Substituting \eqref{eq:RWF_scalar} into \eqref{eq:Rcond_scalar} yields
\eqref{eq:Rcond_closed}.
If \(D<\min_k\sigma_k^2\), then \(\theta=D\) satisfies
\eqref{eq:theta_scalar} and lies strictly below every \(\sigma_k^2\), so
uniqueness forces \(D_k^{\star}=D\) for every \(k\).

\subsection{Scalar Gaussian test channel}
\label{app:gauss_tc}

Fix \(k\) and \(0<D_k^{\star}\le\sigma_k^2\).
The backward channel
\begin{align}
\label{eq:gauss_tc}
    &\widehat{X}
    =
    \mu_k
    +
    \alpha_k(X-\mu_k)
    +
    \nu_k, \\
    &
    \alpha_k
    =
    1-D_k^{\star}/\sigma_k^{2},
    \qquad
    \nu_k
    \sim
    \mathcal{N}(0,D_k^{\star}\alpha_k)\nonumber
\end{align}
independent of \(X\), satisfies
\(\mathbb{E}\bigl[(X-\widehat{X})^2\mid\widetilde{T}=k\bigr]=D_k^{\star}\)
and
\(I(X;\widehat{X}\mid\widetilde{T}=k)=\tfrac12\log(\sigma_k^2/D_k^{\star})\)
\cite{Berger1971,CoverThomas}.
If \(D_k^{\star}\ge\sigma_k^2\), the same formulae hold with the convention
\(\alpha_k=0\) and \(\widehat{X}=\mu_k\).
More precisely, consider a sequence of \(n\)-letter codes for an i.i.d.\
\(\mathcal{N}(\mu_k,\sigma_k^2)\) source whose rates converge to
\(R_k(D_k^{\star})\) and whose per-letter distortions converge to
\(D_k^{\star}\).  If \(\widehat{X}^n\) denotes the induced reconstruction,
then
\[
    \frac{1}{n}I(X^n;\widehat{X}^n)
    =
    R_k(D_k^{\star})+o(1).
\]

\subsection{Vector reverse water-filling}
\label{app:rwf_vec}

Let \(f_k=\mathcal{N}(\boldsymbol{\mu}_k,\boldsymbol{\Sigma}_k)\) with
$$
\boldsymbol{\Sigma}_k
=
\mathbf{U}_k\operatorname{diag}(\lambda_{k,1},\ldots,\lambda_{k,d})
\mathbf{U}_k^{\mathsf{T}}$$
and \(\lambda_{k,m}>0\).
A translation does not affect the MSE RDF, and
\(\mathbf{U}_k^{\mathsf{T}}\) preserves \(\|\cdot\|^2\), so the conditional
problem given \(\widetilde{T}=k\) decomposes into \(d\) independent scalar
Gaussian subproblems along the eigen-directions of \(\boldsymbol{\Sigma}_k\)
\cite{Berger1971,CoverThomas}.
Write \(D_{k,m}\) for the distortion assigned to eigen-direction \(m\) in
component \(k\).
Then
\begin{equation}
    R_{X|T}^{ }(D)
    =
    \inf_{D_{k,m}}
    \frac12
    \sum_{k=1}^{K}
    \pi_k
    \sum_{m=1}^{d}
    \log\frac{\lambda_{k,m}}{D_{k,m}},
\end{equation}
subject to \(\sum_k\pi_k\sum_m D_{k,m}=D\) and
\(0<D_{k,m}\le\lambda_{k,m}\).
(As in Appendix~\ref{app:rwf_scalar}, coordinates with
\(D_{k,m}>\lambda_{k,m}\) may be excluded.)
The program is strictly convex.
A common Lagrange multiplier for the total-distortion constraint yields
the reverse water-filling allocation
\(D_{k,m}^{\star}=\min(\lambda_{k,m},\theta)\) with \(\theta>0\) the unique
solution of \eqref{eq:theta_vec}, where uniqueness follows because the right-hand
side of \eqref{eq:theta_vec} is continuous and strictly increasing from
\(0\) to \(\sum_k\pi_k\operatorname{tr}(\boldsymbol{\Sigma}_k)\) as \(\theta\)
ranges over \([0,\max_{k,m}\lambda_{k,m}]\).
Substitution then produces \eqref{eq:R_GRM_vec}.

If \(D/d<\min_{k,m}\lambda_{k,m}\), the choice \(\theta=D/d\) satisfies
\eqref{eq:theta_vec} and lies strictly below every eigenvalue, so uniqueness
forces \(D_{k,m}^{\star}=D/d\) for all \(k,m\).
A nondegenerate Gaussian in \(\mathbb{R}^d\) meets its MSE SLB at total
distortion \(D\) if and only if this fully active regime holds for that
covariance, i.e., \(D/d\le\lambda_{\min}(\boldsymbol{\Sigma}_k)\).

\subsection{Proof of Lemma~\ref{lem:ex_slb_construction}}
\label{app:ex_slb}

\begin{proof}[Proof of Lemma~\ref{lem:ex_slb_construction}]
Let \(\varphi\) be the standard normal density, and define
\begin{equation}
    \rho(t)
    :=
    \begin{cases}
    \exp\!\bigl(-(1-t^2)^{-1}\bigr), & |t|<1,\\
    0, & |t|\ge1,
    \end{cases}
\end{equation}
which is a nonnegative smooth function supported on \([-1,1]\).
Fix \(w\in(0,1/4)\), set \(\rho_w(x):=\rho(x/w)\), and define
\begin{align}
    \widetilde{\psi}(x)
    &:=
    -4\bigl(\rho_w(x-1)+\rho_w(x+1)\bigr)
    \nonumber\\ &\qquad+
    \bigl(\rho_w(x-2)+\rho_w(x+2)\bigr)
    +
    6\,\rho_w(x).
\end{align}
This function is even, smooth, and supported on \([-2-w,2+w]\). The five
translated functions in its definition have disjoint supports because
\(w<1/4\). Write
\[
    I_0:=\int_{\mathbb{R}}\rho_w(x)\,dx,
    \qquad
    I_2:=\int_{\mathbb{R}}x^2\rho_w(x)\,dx.
\]
Since \(\rho_w\) is even, \(\int x\rho_w(x)\,dx=0\). Moreover,
\(\int(x-a)^2\rho_w(x-a)\,dx=I_2+a^2 I_0\) yields
\begin{align}
    \int\widetilde{\psi}
    &=
    \bigl(6-4-4+1+1\bigr)I_0
    =
    0,
    \nonumber\\
    \int x\widetilde{\psi}
    &=
    0,
    \nonumber\\
    \int x^2\widetilde{\psi}
    &=
    \bigl(6-4-4+1+1\bigr)I_2\\
    & \qquad+
    \bigl(-4\cdot1^2-4\cdot1^2+2^2+2^2\bigr)I_0\\
    &=
    0.
\end{align}
where the first-moment identity also follows from the evenness of
\(\widetilde{\psi}\).

The function \(\widetilde{\psi}\) is bounded, and \(\varphi\) is
continuous and strictly positive on its compact support. We may therefore
choose \(\varepsilon>0\) sufficiently small that
\(\psi:=\varepsilon\widetilde{\psi}\) satisfies
\(|\psi|\le\tfrac12\varphi\) pointwise.
Set \(f_1:=\varphi+\psi\) and \(f_2:=\varphi-\psi\).
Then \(f_k\ge\tfrac12\varphi>0\) and \(\int f_k=1\). The moment
identities give
\[
    \int x f_k(x)\,dx=0,
    \qquad
    \int x^2 f_k(x)\,dx=1,
    \qquad k\in\{1,2\}.
\]
Moreover, \(f_1\neq f_2\) because \(\psi\not\equiv0\), and
\((f_1+f_2)/2=\varphi\). Since \(f_k=\varphi\) outside a compact set,
both densities have finite differential entropy.

It remains to show that neither density has the Gaussian-convolution
representation stated in the lemma. For every \(D>0\) and every
probability measure \(\nu\), the convolution \(\nu*\varphi_D\) is real
analytic. To see this, consider its extension
\[
    g(z):=
    \int_{\mathbb{R}}
    \frac{1}{\sqrt{2\pi D}}
    \exp\!\left(-\frac{(z-y)^2}{2D}\right)\nu(dy)
\]
to \(z\in\mathbb{C}\). If \(z=a+ib\), the magnitude of the integrand is
bounded by
\[
    \frac{1}{\sqrt{2\pi D}}\exp\!\left(\frac{b^2}{2D}\right).
\]
This bound is uniform in \(y\) and locally uniform in \(z\), so
differentiation under the integral is valid and \(g\) is analytic.
Consequently, if \(f_k=\nu*\varphi_D\), then \(f_k\) and
\(f_k-\varphi=\pm\psi\) would be real analytic. The function \(\psi\),
however, is zero outside a compact interval and is not identically zero,
which is impossible for a real-analytic function on \(\mathbb{R}\).
Thus neither \(f_1\) nor \(f_2\) has the required representation.
\end{proof}

\color{black}
\section{Variational and Rare-Binary Arguments for Zero Rate}
\label{app:zero}

\subsection{Variational Preliminaries}

We first verify the properties of \(\tau^2\) used in
Theorem~\ref{thm:zero_rate_vec}. Define
\[
    \sigma_*^2:=\max_k\lambda_{\max}(\boldsymbol{\Sigma}_k),
    \qquad
    b:=\max_k\|\boldsymbol{\mu}_k\|.
\]
Because \(\mathbb{E}[\boldsymbol{\mu}_T]=\mathbf{0}\), Hoeffding's lemma
and the conditional Gaussian moment generating functions give
\begin{align}
    \log\mathbb{E}\bigl[e^{\mathbf{u}^{\mathsf{T}}\mathbf{X}}\bigr]
    &\le
    \tfrac12\sigma_*^2\|\mathbf{u}\|^2
    +
    \log\mathbb{E}
    \bigl[e^{\mathbf{u}^{\mathsf{T}}\boldsymbol{\mu}_T}\bigr]
    \nonumber\\
    &\le
    \tfrac12(\sigma_*^2+b^2)\|\mathbf{u}\|^2.
\end{align}
Thus \(\tau^2<\infty\). Expansion at the origin also gives
\[
    \tau^2
    \ge
    \lambda_{\max}(\boldsymbol{\Sigma}_{\mathbf{X}})
    >0.
\]

Let \(Q\ll P_{\mathbf{X}}\) have finite relative entropy and write
\(\mathbf{m}_Q:=\mathbb{E}_Q[\mathbf{X}]\). This mean is well defined:
for any \(c>0\), the entropy inequality gives
\[
    c\mathbb{E}_Q\|\mathbf{X}\|
    \le
    D_{\mathrm{KL}}(Q\Vert P_{\mathbf{X}})
    +
    \log\mathbb{E}[e^{c\|\mathbf{X}\|}]
    <\infty.
\]
For \(\mathbf{u}\in\mathbb{R}^d\), define
\[
    \frac{dP_{\mathbf{u}}}{dP_{\mathbf{X}}}(\mathbf{x})
    :=
    \frac{e^{\mathbf{u}^{\mathsf{T}}\mathbf{x}}}
    {\mathbb{E}[e^{\mathbf{u}^{\mathsf{T}}\mathbf{X}}]}.
\]
The change-of-measure identity and
\(D_{\mathrm{KL}}(Q\Vert P_{\mathbf{u}})\ge0\) imply
\[
    \mathbf{u}^{\mathsf{T}}\mathbf{m}_Q
    \le
    D_{\mathrm{KL}}(Q\Vert P_{\mathbf{X}})
    +
    \tfrac12\tau^2\|\mathbf{u}\|^2.
\]
Taking \(\mathbf{u}=\mathbf{m}_Q/\tau^2\) proves
\begin{equation}
\label{eq:tau_kl}
    \|\mathbf{m}_Q\|^2
    \le
    2\tau^2D_{\mathrm{KL}}(Q\Vert P_{\mathbf{X}}).
\end{equation}

It remains to establish the variational form used for the binary upper
bound. Set
\[
    S:=
    \sup_{\substack{
        Q\ll P_{\mathbf{X}}:\\
        0<D_{\mathrm{KL}}(Q\Vert P_{\mathbf{X}})<\infty}}
    \frac{\|\mathbb{E}_Q[\mathbf{X}]\|^2}
    {2D_{\mathrm{KL}}(Q\Vert P_{\mathbf{X}})}.
\]
Inequality \eqref{eq:tau_kl} gives \(S\le\tau^2\).
Moreover, \(S>0\): if \(\mathbf{v}\) is a unit eigenvector of
\(\boldsymbol{\Sigma}_{\mathbf{X}}\) with eigenvalue \(\lambda>0\), then,
as \(t\to0\),
\[
    \mathbb{E}_{P_{t\mathbf{v}}}[\mathbf{X}]
    =
    t\boldsymbol{\Sigma}_{\mathbf{X}}\mathbf{v}+o(t),
    \qquad
    D_{\mathrm{KL}}(P_{t\mathbf{v}}\Vert P_{\mathbf{X}})
    =
    \tfrac12t^2\lambda+o(t^2),
\]
so the corresponding ratio converges to \(\lambda\).

Conversely, the Donsker-Varadhan variational formula
\cite{DonskerVaradhan1975} gives
\[
    \log\mathbb{E}\bigl[e^{\mathbf{u}^{\mathsf{T}}\mathbf{X}}\bigr]
    =
    \sup_{\substack{
        Q\ll P_{\mathbf{X}}:\\
        D_{\mathrm{KL}}(Q\Vert P_{\mathbf{X}})<\infty}}
    \left\{
        \mathbf{u}^{\mathsf{T}}\mathbb{E}_Q[\mathbf{X}]
        -
        D_{\mathrm{KL}}(Q\Vert P_{\mathbf{X}})
    \right\}.
\]
The definition of \(S\), including the zero-divergence case
\(Q=P_{\mathbf{X}}\), therefore yields
\begin{align}
    \log\mathbb{E}\bigl[e^{\mathbf{u}^{\mathsf{T}}\mathbf{X}}\bigr]
    &\le
    \sup_{\mathbf{m}\in\mathbb{R}^d}
    \left\{
        \mathbf{u}^{\mathsf{T}}\mathbf{m}
        -
        \frac{\|\mathbf{m}\|^2}{2S}
    \right\}
    \nonumber\\
    &=
    \tfrac12S\|\mathbf{u}\|^2.
\end{align}
Hence \(\tau^2\le S\), proving
\begin{equation}
\label{eq:tau_var}
    \tau^2
    =
    \sup_{\substack{
        Q\ll P_{\mathbf{X}}:\\
        0<D_{\mathrm{KL}}(Q\Vert P_{\mathbf{X}})<\infty}}
    \frac{\|\mathbb{E}_Q[\mathbf{X}]\|^2}
    {2D_{\mathrm{KL}}(Q\Vert P_{\mathbf{X}})}.
\end{equation}

\subsection{Rare-Binary Mutual Information}

Let \(Q\ll P_{\mathbf{X}}\) have likelihood ratio
\(L=dQ/dP_{\mathbf{X}}\) satisfying
\(\|L\|_\infty\le B<\infty\), and write
\[
    \mathbf{m}:=\mathbb{E}_Q[\mathbf{X}],
    \qquad
    K_Q:=D_{\mathrm{KL}}(Q\Vert P_{\mathbf{X}}).
\]
For \(0<q<\min\{1,B^{-1}\}\), define
\(P(J=1\mid\mathbf{X})=qL\). Direct integration gives
\(P(J=1)=q\), \(P_{\mathbf{X}\mid J=1}=Q\), and
\begin{equation}
    r_{0,q}
    :=
    \frac{dP_{\mathbf{X}\mid J=0}}{dP_{\mathbf{X}}}
    =
    \frac{1-qL}{1-q}.
\end{equation}
Consequently,
\begin{equation}
\label{eq:binary_mi_decomp}
    I(\mathbf{X};J)
    =
    qK_Q
    +
    (1-q)\mathbb{E}\bigl[r_{0,q}\log r_{0,q}\bigr].
\end{equation}
Set
\[
    \eta_q:=r_{0,q}-1
    =
    \frac{q(1-L)}{1-q}.
\]
Then \(\mathbb{E}[\eta_q]=0\) and
\(\|\eta_q\|_\infty=O(q)\) as \(q\downarrow0\). The Taylor expansion
\[
    (1+t)\log(1+t)
    =
    t+\tfrac12t^2+O(|t|^3)
\]
holds uniformly over the range of \(\eta_q\) for all sufficiently small
\(q\). Moreover,
\[
    \mathbb{E}[\eta_q^2]
    =
    \frac{q^2}{(1-q)^2}
    \mathbb{E}[(L-1)^2]
    =
    \frac{q^2}{(1-q)^2}
    \chi^2(Q\Vert P_{\mathbf{X}}).
\]
Substitution into \eqref{eq:binary_mi_decomp} therefore gives
\begin{equation}
    I(\mathbf{X};J)
    =
    qK_Q
    +
    \tfrac12q^2\chi^2(Q\Vert P_{\mathbf{X}})
    +
    O(q^3).
\end{equation}

\subsection{Truncation}

Let \(Q\ll P_{\mathbf{X}}\) satisfy
\(K_Q:=D_{\mathrm{KL}}(Q\Vert P_{\mathbf{X}})<\infty\), and let
\(L=dQ/dP_{\mathbf{X}}\). Since a finite GMM has
\(\mathbb{E}[e^{c\|\mathbf{X}\|}]<\infty\) for every finite \(c>0\), the
entropy inequality implies
\(\mathbb{E}_Q\|\mathbf{X}\|<\infty\).
For \(n\ge1\), define
\[
    a_n:=\mathbb{E}[L\wedge n],
    \qquad
    L_n:=\frac{L\wedge n}{a_n},
    \qquad
    \frac{dQ_n}{dP_{\mathbf{X}}}:=L_n.
\]
Then \(a_n\uparrow1\), each \(L_n\) is bounded, and
\[
    \mathbb{E}_{Q_n}[\mathbf{X}]
    =
    \frac{\mathbb{E}[\mathbf{X}(L\wedge n)]}{a_n}
    \longrightarrow
    \mathbb{E}[\mathbf{X}L]
    =
    \mathbb{E}_Q[\mathbf{X}]
\]
by dominated convergence, because
\(\|\mathbf{X}\|(L\wedge n)\le\|\mathbf{X}\|L\) and
\(\mathbb{E}[\|\mathbf{X}\|L]<\infty\).

For the relative entropy,
\begin{equation}
\label{eq:truncated_kl}
    D_{\mathrm{KL}}(Q_n\Vert P_{\mathbf{X}})
    =
    \frac{
        \mathbb{E}[(L\wedge n)\log(L\wedge n)]
    }{a_n}
    -
    \log a_n.
\end{equation}
On \(\{L<1\}\), the integrand in the numerator is unchanged for
\(n\ge1\); on \(\{L\ge1\}\), it increases pointwise to \(L\log L\).
The negative part of \(L\log L\) is bounded and its positive part is
integrable because \(K_Q<\infty\). Monotone convergence on
\(\{L\ge1\}\), together with \(a_n\to1\), therefore shows from
\eqref{eq:truncated_kl} that
\[
    D_{\mathrm{KL}}(Q_n\Vert P_{\mathbf{X}})
    \longrightarrow
    D_{\mathrm{KL}}(Q\Vert P_{\mathbf{X}}).
\]
Thus, whenever
\[
    \frac{\|\mathbb{E}_Q[\mathbf{X}]\|^2}
    {2D_{\mathrm{KL}}(Q\Vert P_{\mathbf{X}})}
    >
    \tau^2-\frac{\varepsilon}{2},
\]
a sufficiently large \(n\) gives a bounded-likelihood-ratio measure
\(Q_n\) whose corresponding ratio is at least
\(\tau^2-\varepsilon\), as used in
\eqref{eq:bounded_Q_ratio}.

\section{Centroid Reduction and Finite-Alphabet Succession}
\label{app:centroid}

\begin{proof}[Proof of Lemma~\ref{lem:centroid}]
Let \(J\) be a random variable taking values in a finite set
\(\mathcal{J}\), jointly distributed with \(\mathbf{X}\).
Mutual information \(I(\mathbf{X};J)\) is a functional of the joint law
\(P_{\mathbf{X},J}\) only, and does not depend on the codebook
\(\{\mathbf{a}_j\}\).
Thus, for any fixed kernel \(P_{J\mid\mathbf{X}}\), it is enough
to minimize the distortion over the codebook.

Since \(\mathbb{E}\|\mathbf{X}\|^2<\infty\), the random vector
\(\mathbf{X}\) lies in \(L^2\), hence
\(\mathbb{E}[\mathbf{X}\mid J]\) exists as an \(L^2\) random vector.
Write \(\mathbf{m}_J:=\mathbb{E}[\mathbf{X}\mid J]\) and
\(\mathbf{m}_j:=\mathbb{E}[\mathbf{X}\mid J=j]\) whenever
\(P(J=j)>0\).
For every codebook,
\begin{align}
    \bigl\|\mathbf{X}-\mathbf{a}_J\bigr\|^2
    &=
    \bigl\|\mathbf{X}-\mathbf{m}_J\bigr\|^2
    +
    \bigl\|\mathbf{m}_J-\mathbf{a}_J\bigr\|^2
    \nonumber\\
    &\qquad
    +
    2\bigl(\mathbf{X}-\mathbf{m}_J\bigr)^{\mathsf{T}}
    \bigl(\mathbf{m}_J-\mathbf{a}_J\bigr).
\end{align}
Taking the conditional expectation given \(J\) kills the cross term,
because \(\mathbb{E}\bigl[\mathbf{X}-\mathbf{m}_J\mid J\bigr]=\mathbf{0}\)
almost surely.
Therefore
\begin{equation}
\label{eq:centroid_pythagoras}
    \mathbb{E}\bigl\|\mathbf{X}-\mathbf{a}_J\bigr\|^2
    =
    \mathbb{E}\bigl\|\mathbf{X}-\mathbf{m}_J\bigr\|^2
    +
    \mathbb{E}\bigl\|\mathbf{m}_J-\mathbf{a}_J\bigr\|^2.
\end{equation}
The second summand is nonnegative, and vanishes if and only if
\(\mathbf{a}_J=\mathbf{m}_J\) almost surely, i.e., if and only if
\(\mathbf{a}_j=\mathbf{m}_j\) for every \(j\) with \(P(J=j)>0\).
On the event \(\{J=j\}\) with \(P(J=j)=0\) the value of \(\mathbf{a}_j\)
does not appear in \eqref{eq:centroid_pythagoras}; such labels may be
deleted without increasing cardinality.
When \(\mathbf{a}_J=\mathbf{m}_J\), the second summand vanishes, and
expanding the first together with
\(\mathbb{E}\bigl[\mathbf{X}^{\mathsf{T}}\mathbf{m}_J\bigr]
=\mathbb{E}\|\mathbf{m}_J\|^2\) yields \eqref{eq:centroid_mse}.

Consequently, replacing an arbitrary codebook by the centroids
\(\mathbf{a}_j=\mathbb{E}[\mathbf{X}\mid J=j]\) leaves \(I(\mathbf{X};J)\)
unchanged and does not increase distortion.
Every SSQ test channel of size at most \(M\) is thus dominated by one with the same
conditional law and centroid reconstruction, and the infimum
\eqref{eq:rdf_M} may be restricted to kernels \(P_{J\mid\mathbf{X}}\)
with \(|\mathcal{J}|\le M\) and
\(\mathbb{E}\bigl\|\mathbf{X}-\mathbb{E}[\mathbf{X}\mid J]\bigr\|^2\le D\).
\end{proof}

\begin{proof}[Proof of Lemma~\ref{lem:RM_limit}]
The constraint set in \eqref{eq:rdf_M} enlarges with \(M\), and is contained
in that of \eqref{eq:rdf}, which gives \eqref{eq:RM_mono}.
The limit \eqref{eq:RM_limit} is the finite-output-alphabet approximation
of Finamore and Pearlman \cite{FinamorePearlman1980}; we record a short
clipping argument for completeness.
Fix \(D\in(0,D_{\max})\) and \(\varepsilon>0\). Continuity of \(R^{ }\) on
\((0,D_{\max})\) supplies \(\eta\in(0,D)\) and a (possibly continuous) test
channel of distortion \(D-\eta\) with mutual information at most
\(R^{ }(D)+\varepsilon\).
Clip the reproduction to a cube of side \(A\) large enough that the clipping
distortion is at most \(\eta/2\), then quantize the cube by a cubic lattice
of mesh \(\delta\) small enough that the additional MSE is at most
\(\eta/2\); both choices are possible because \(\mathbb{E}\|\mathbf{X}\|^2<\infty\).
The resulting reproduction takes finitely many values, say \(M=M(A,\delta)\).
Data processing does not increase mutual information, and replacing each
lattice point by the conditional centroid cannot increase distortion
(Lemma~\ref{lem:centroid}).
Hence \(R_M^{ }(D)\le R^{ }(D)+\varepsilon\).
Monotonicity \eqref{eq:RM_mono} yields \eqref{eq:RM_limit}.
\end{proof}

\section{Proof of the Strict Inequality Step of Corollary~\ref{cor:gaps}}
\label{app:gaps}

Fix \(M<\infty\) and \(D\in(0,D_{\max})\) such that
\(R_M^{ }(D)<\infty\), and suppose that every reproduction law attaining
\(R^{ }(D)\) is non-atomic. We show that
\(R_M^{ }(D)>R^{ }(D)\).

Write \(Y:=\widehat{\mathbf{X}}\) for a reconstruction.
After merging coinciding codewords one may assume
\(I(\mathbf{X};J)=I(\mathbf{X};Y)\) in \eqref{eq:rdf_M}, so it is enough
to argue with laws of \(Y\) supported on at most \(M\) atoms.

Suppose, toward a contradiction, that equality holds.
Then there exist \(M\)-atomic random vectors \(Y_n\) with
\(\mathbb{E}\|\mathbf{X}-Y_n\|^2\le D\) and
\(I(\mathbf{X};Y_n)\to R^{ }(D)\).
The second-moment bound
\(\mathbb{E}\|Y_n\|^2\le 2D_{\max}+2D\) yields tightness of
\(\{P_{Y_n}\}\) and of the joints \(\{P_{\mathbf{X},Y_n}\}\) on
\(\mathbb{R}^{d}\times\mathbb{R}^{d}\).
Passing to a subsequence, \(P_{\mathbf{X},Y_n}\Rightarrow P_{\mathbf{X},Y}\)
for some \(Y\).
The set of Borel probabilities with at most \(M\) atoms is weakly closed
along tight sequences: masses of escaping atoms vanish, and the remaining
locations have a convergent subsequence.
Thus \(Y\) is \(M\)-atomic.
The integrand \(\|\mathbf{x}-\mathbf{y}\|^2\) is nonnegative and
lower semicontinuous, so
\(\mathbb{E}\|\mathbf{X}-Y\|^2
\le\liminf_{n\to\infty}\mathbb{E}\|\mathbf{X}-Y_n\|^2\le D\).
The Donsker--Varadhan variational representation
\cite{DonskerVaradhan1975} gives
\[
D_{\mathrm{KL}}(\mu\Vert\nu)
=
\sup_{f\in C_b(\mathbb{R}^{2d};\mathbb{R})}
\Bigl(
\int f\,d\mu-\log\int e^{f}\,d\nu
\Bigr).
\]
Since \(P_{Y_n}\Rightarrow P_Y\), one has
\(P_{\mathbf{X}}\otimes P_{Y_n}\Rightarrow P_{\mathbf{X}}\otimes P_Y\).
Mutual information is the relative entropy of the joint against the
product of the marginals, so
\(I(\mathbf{X};Y)\le\liminf_{n\to\infty} I(\mathbf{X};Y_n)=R^{ }(D)\).
Feasibility of the limit pair forces equality, so \(Y\) attains
\(R^{ }(D)\).
This contradicts the assumed non-atomicity.

\section{Proof of Lemma~\ref{lem:floor}}
\label{app:floor}

\begin{proof}[Proof of Lemma~\ref{lem:floor}]
For any \(P_{J|\mathbf{X}}\) and any codebook,
\begin{align}
    &\mathbb{E}\|\mathbf{X}-\mathbf{a}_J\|^2\nonumber\\
    =&
    \mathbb{E}
    \Bigl[
    \sum_{j=1}^{M}
    P(J=j|\mathbf{X})\,\|\mathbf{X}-\mathbf{a}_j\|^2
    \Bigr]
    \ge
    \mathbb{E}
    \Bigl[
    \min_{1\le j\le M}
    \|\mathbf{X}-\mathbf{a}_j\|^2
    \Bigr],
\end{align}
so no (soft or hard) \(M\)-ary scheme beats \(D_M^{\star }\).

For positivity, a nondegenerate GMM density is continuous and bounded below
by some \(c>0\) on a closed ball \(B\) of radius \(r>0\) and volume
\(v>0\).
For an arbitrary codebook, the Voronoi cells partition \(B\) into at most
\(M\) measurable pieces of volumes \(v_j\) with \(\sum_j v_j=v\).
On a cell of volume \(v_j\),
\(\int\min_a\|\mathbf{x}-\mathbf{a}\|^2\,d\mathbf{x}\ge\kappa_d v_j^{1+2/d}\)
for a dimensional constant \(\kappa_d>0\) (the minimum is the second moment
about the centroid, and is homogeneous of degree \(1+2/d\)).
Convexity of \(t\mapsto t^{1+2/d}\) yields
\(\sum_j v_j^{1+2/d}\ge M\cdot(v/M)^{1+2/d}=v^{1+2/d}M^{-2/d}\), whence
\(D_M^{\star }\ge c\kappa_d v^{1+2/d}M^{-2/d}>0\).
The scalar case \(d=1\) recovers the elementary bound
\(D_M^\star\ge c\ell^3/(12M^2)\) on an interval of length \(\ell\).

For attainment, let \(\{\mathbf{a}_j^\star\}\) nearly minimize
\eqref{eq:floor_expr} and set
\(J^\star=\arg\min_j\|\mathbf{X}-\mathbf{a}_j^\star\|^2\).
Then \(I(\mathbf{X};J^\star)=H(J^\star)\le\log M\) and the distortion can be
made arbitrarily close to the floor, so
\(R_M^{ }(D_M^{\star })\le\log M\).
For larger rate the feasible set still cannot beat \(D_M^{\star }\),
hence saturation.
Items 1--3 follow.
For item 4, a nondegenerate GMM density is bounded, strictly positive on a
nonempty open set, and has finite moments of every order, so
\(0<\int f_{\mathbf{X}}^{d/(d+2)}<\infty\).
Zador's theorem \cite{Zador1982} therefore gives
\[
\lim_{M\to\infty}
M^{2/d}\,D_M^{\star }
=
G_d
\int_{\mathbb{R}^d}
f_{\mathbf{X}}(\mathbf{x})^{d/(d+2)}\,d\mathbf{x}
\in(0,\infty),
\]
which is \eqref{eq:zador}.
In dimension one, \(G_1=1/12\) \cite{PanterDite1951}.
\end{proof}

\color{black}

\section{Blahut-Arimoto / Deterministic Annealing Updates}
\label{app:ba}

SSQ operating points used in Section~\ref{sec:numerics}
are computed by the following standard free-energy iteration.
\color{black}

Let \(J\in\{1,\ldots,M\}\) with soft assignments
\begin{align}
    p_j(\mathbf{x})
    &:=
    P(J=j|\mathbf{X}=\mathbf{x}),
    \qquad
    \\q_j
    &:=
    P(J=j)
    =
    \int f_{\mathbf{X}}(\mathbf{x})\,p_j(\mathbf{x})\,d\mathbf{x},
\end{align}
and reconstruction \(\widehat{\mathbf{X}}=\mathbf{a}_J\).
The associated distortion and mutual information are
\begin{align}
    D
    &=
    \int f_{\mathbf{X}}(\mathbf{x})
    \sum_{j=1}^{M}
    p_j(\mathbf{x})\,\|\mathbf{x}-\mathbf{a}_j\|^2\,d\mathbf{x},
    \\
    I(\mathbf{X};J)
    &=
    \int f_{\mathbf{X}}(\mathbf{x})
    \sum_{j=1}^{M}
    p_j(\mathbf{x})\log\frac{p_j(\mathbf{x})}{q_j}\,d\mathbf{x}.
\end{align}
For \(\lambda>0\), define the Lagrangian
\begin{equation}
    \mathcal{L}
    :=
    D+\lambda I(\mathbf{X};J).
\end{equation}
Since the masses $q_j$ are constrained by the assignments
$p_j(\cdot)$, it is convenient to introduce an auxiliary probability vector
\(\mathbf{q}=(q_1,\ldots,q_M)\) and the free-energy functional
\begin{align}
    F(\mathbf{p},\mathbf{q},\mathbf{a})
    &=
    \int f_{\mathbf{X}}(\mathbf{x})
    \sum_{j=1}^{M}
    p_j(\mathbf{x})\,\|\mathbf{x}-\mathbf{a}_j\|^2\,d\mathbf{x}
    \nonumber\\
    &\quad
    +
    \lambda
    \int f_{\mathbf{X}}(\mathbf{x})
    \sum_{j=1}^{M}
    p_j(\mathbf{x})\log\frac{p_j(\mathbf{x})}{q_j}\,d\mathbf{x}.
\end{align}
For every fixed conditional assignment \(\mathbf{p}\),
\begin{equation}
    I(\mathbf{X};J)
    =
    \min_{\mathbf{q}}
    \int f_{\mathbf{X}}(\mathbf{x})
    \sum_{j=1}^{M}
    p_j(\mathbf{x})\log\frac{p_j(\mathbf{x})}{q_j}\,d\mathbf{x},
\end{equation}
and the unique minimizer is the induced marginal
\(q_j=\int f_{\mathbf{X}}(\mathbf{x})\,p_j(\mathbf{x})\,d\mathbf{x}\).
Consequently, joint minimization of \(F\) is equivalent to minimization of
\(\mathcal{L}\).

Stationarity of \(F\) yields the Gibbs assignment, mass, and centroid
equations
\begin{equation}
    p_j(\mathbf{x})
    =
    \frac{
    q_j\exp\bigl(-\|\mathbf{x}-\mathbf{a}_j\|^2/\lambda\bigr)
    }{
    \sum_{\ell=1}^{M}
    q_\ell\exp\bigl(-\|\mathbf{x}-\mathbf{a}_\ell\|^2/\lambda\bigr)
    },
\end{equation}
\begin{equation}
    q_j
    =
    \int f_{\mathbf{X}}(\mathbf{x})\,p_j(\mathbf{x})\,d\mathbf{x},
    \qquad
    \mathbf{a}_j
    =
    \mathbb{E}[\mathbf{X}|J=j].
\end{equation}
In practice, an index with \(q_j=0\) may be dropped. Alternating between these updates gives rise to a BA / deterministic-annealing
iteration solution\cite{Blahut1972,Rose1994}.
We note that the joint problem over
\(\mathbf{a}_j\) remains non-convex and there is no practical guarantee of a global optimum.

\section{Notation}
\label{app:notation}
Table~\ref{tab:notation} summarizes the main notation.

\begin{table}[!t]
\renewcommand{\arraystretch}{1.15}
\caption{Notation used throughout the paper.}
\label{tab:notation}
\centering
\begin{tabular}{@{}ll@{}}
\toprule
Symbol & Meaning \\
\midrule
\(K\) & number of mixture components \\
\(M\) & SSQ codebook size \\
\(d\) & source dimension \\
\(D_{\max}\) & zero-rate distortion \(\operatorname{tr}(\boldsymbol{\Sigma}_{\mathbf{X}})\) \\
\(R^{ }(D),R_M^{ }(D)\) & Shannon / SSQ RDF \\
\(\lambda_{k,m}\) & eigenvalues of \(\boldsymbol{\Sigma}_k\) \\
\(\lambda_{\min}(\mathbf{A}),\lambda_{\max}(\mathbf{A})\) & extreme eigenvalues of \(\mathbf{A}\) \\
\(\tau^2\) & directional sub-Gaussian proxy \\
\(T,\widetilde{T}\) & mixture label / posterior sample \\
\(J\) & SSQ reconstruction index \\
\(R_{\mathrm{GRM}}(D)\) & GRM achievable rate \\
GRM & Gibbs-relabeling method \\
SQ & scalar quantization \\
SSQ & stochastic scalar quantization \\
\bottomrule
\end{tabular}
\end{table}

\balance
\bibliographystyle{IEEEtran}
\bibliography{ref}

\end{document}